\documentclass[11pt]{article}
\usepackage[all=normal,paragraphs,wordspacing,tracking]{savetrees}
\usepackage[text={6.5in,9in},centering]{geometry}

\usepackage{booktabs}
\usepackage{thmtools}
\usepackage{thm-restate}
\usepackage{amsfonts}
\usepackage{amssymb}
\usepackage{amsmath}
\usepackage{tikz}

\usepackage{relsize}
\usepackage[capitalise]{cleveref}
\usepackage{amsthm} 
\usepackage{enumitem}
\newcommand{\pBel}{{\boldsymbol{\beta}}}

 \newtheorem{example}{Example}
 \newtheorem{theorem}{Theorem}[section]
\newtheorem{definition}[theorem]{Definition}
 
 \newtheorem{lemma}[theorem]{Lemma}

\usepackage{thmtools}
\usepackage{thm-restate}

\newcommand{\iInNEnv} {$i\in \Proc\cup \{e\}$}

\newcommand{\li}{{\ell_i}}
\newcommand{\lip}{\ell'_i}
\newcommand{\Li}{{L_i}}
\newcommand{\lEnv}{\ell_e}

\newcommand{\rit}{r_i(t)}
\newcommand{\Proti}{P_i}
\newcommand{\Qli}{Q^\li}

\newcommand{\cond}{\varphi}
\newcommand{\act}{\alpha}
\newcommand{\KoP}{KoP}
\newcommand{\Act}{\mathsf{Act}}
\newcommand{\Acti}{\Act_i}

\newcommand{\expectation}{\mathbb{E}}

\newcommand{\acti}{\mathsf{does}_i(\act)}
\newcommand{\newacti}{\mathsf{does}_i(\act)}

\newcommand{\sat}{\models}

\newcommand{\red}[1]{\textcolor{red}{#1}}

\newcommand{\given}{\,\boldsymbol{\vert}\,}
\newcommand{\Bigiven}{\,\boldsymbol{\big{|}}}

\newcommand{\go}{\mathtt{go}}

\newcommand{\fireA}{\textsf{fire}_A}
\newcommand{\fireB}{\textsf{fire}_B}

\newcommand{\FS}{\mathsf{FS}}
\newcommand{\Proc}{\textsf{Ags}}
\newcommand{\timei}{\mathtt{time}_i}
\newcommand{\rootv}{\lambda}
\newcommand{\RT}{{R_T}}
\newcommand{\XT}{{{\mathcal{X}}_T}}
\newcommand{\Pts}{\mathtt{Pts}}
\newcommand{\ract}{\underline{\act}}
\newcommand{\racti}{\underline{\act}} 

\newcommand{\rli}{\underline{\ell}_i}

\newcommand{\muT}{\mu_{T}}
\newcommand{\defemph}[1]{\textbf{\textit{#1}}}
\newcommand{\eqdef}{\triangleq}
\newcommand{\Lacti}{L_{\act}}
\newcommand{\gvn}{\;\!|\;\!}
\newcommand{\Ract}{R_{\racti}}  
\newcommand{\vareps}{\varepsilon}
\newcommand{\Pro}[1]{\muT\big({\cdot\!\cdot\!\cdot#1}\,1-\delta\Bigiven{\cdot}\big)}
\newcommand{\bit}{\mathsf{bit}}
\newcommand{\marg}[1]{\red{$\boldsymbol{\odot}$\marginpar{\red{#1}}}}

\newcommand{\Spec}{{\sf Spec}}
\newcommand{\both}{\varphi_{\scriptscriptstyle{both}}}
\newcommand{\YES}{\mbox{{\it `}$\,\mathtt{Yes}\!${\it '}}}
\newcommand{\NO}{\mbox{{\it `}$\,\mathtt{No}\!${\it '}}}

\title{Probably Approximately Knowing}

\author{
  Nitzan Zamir\\
  \texttt{snussa@cs.technion.ac.il}
  \and
  Yoram Moses\footnote{This work was funded in part by ISF grant 2061/19.
Yoram Moses is the Israel Pollak academic chair at the Technion. This is a full version of a paper whose extended abstract appears in the proceeding of PODC~2020.}\\
  \texttt{moses@ee.technion.ac.il}
}

\begin{document}

\date{}
\maketitle


\begin{abstract}

Whereas deterministic protocols are typically guaranteed to obtain particular goals of interest, probabilistic protocols typically provide only probabilistic guarantees. 
This paper initiates an investigation 
of the interdependence between actions and subjective beliefs of agents in a probabilistic setting. In particular, we study what probabilistic beliefs an agent should have when performing actions, in a protocol that satisfies a  probabilistic constraint of the form: {\it Condition~$\varphi$ should hold with probability at least~$p$ when action~$\act$ is performed}.  
Our main result is that the {\it expected degree of an agent's belief} in~$\varphi$ when it performs~$\act$ equals the probability that~$\varphi$ holds when $\act$ is performed.  Indeed, if the threshold of the probabilistic constraint should hold with probability $p=1-\vareps^2$ for some small value of~$\vareps$ then, with probability $1-\vareps$, when the agent acts it will assign a probabilistic belief no smaller than~$1-\vareps$ to the possibility that~$\varphi$ holds. 
In other words, viewing strong belief as, intuitively, approximate knowledge, the agent must \defemph{probably approximately know} (PAK-know) that~$\varphi$ is true when it acts. 

\end{abstract}

\section{Introduction}
\label{sec:Intro}
One of the defining features of distributed and multi-agent systems is that an agent's actions can only depend on its local information.
While this local information cannot typically contain a complete description of the state of the system, it must still  be sufficiently rich to support the actions that the agent takes.  
Thus, for example, in runs of a mutual exclusion (ME) protocol, an agent that enters the critical section  must know (based on its local state) that no other agent is in the critical section (cf.\ \cite{DijkstraME,ChandyMisra}). 
In probabilistic protocols, or even in deterministic protocols that operate in a probabilistic setting, the constraints on actions are often specified in probabilistic, rather than absolute, terms. One could, for example, consider a probabilistic requirement stating that upon entry to the critical section, it  should be empty with very high probability, rather than in all cases. 
In this case, the connection between an agent's actions and its local information is apparently not as tight. This paper initiates an investigation of the connection between the two in protocols that satisfy probabilistic constraints.
The following example provides a taste of the subject matter and will serve us to discuss some of the issues involved:  
\begin{example}
\label{ex:1}
We are given a synchronous message-passing system with two agents, Alice and Bob. At any given round, each agent can send messages to the other, and can either perform a ``firing'' action ($\fireA$ and  $\fireB$) or skip. 
Communication between them is unreliable, with every message sent being lost with probability 0.1, and being delivered in the round in which it is sent with probability~\mbox{0.9}. No message is delivered late, and probabilities for different messages are independent. Both agents begin operation at time~0, and Alice is assumed to have a binary variable ``$\go$'' in her initial state, whose value is~0 with probability~0.5, and is~1 otherwise. Given the unreliability of communication, it is not possible to ensure that both agents will always fire simultaneously. Instead, we consider the {\em relaxed firing squad problem}, in which they do so with high probability: 
\begin{itemize}
    \item[\Spec{\sf:}] If $\go=0$ then neither agent ever fires; while \\if $\go=1$ they should attempt to coordinate a joint firing. In particular, \\
    The probability that both agents fire, given that Alice  is firing, should be at least 0.95. 
\end{itemize}
Now consider the following protocol, $\FS$, in which, when $go=1$ Alice sends two messages to Bob in the first round, and fires at time~2 (in the third round). When $\go=0$ she sends no messages and never fires. 
Bob acts as follows: If he receives at least one message from Alice in the first round, he sends her a $\YES$ message in the second round and fires in the third. If, however, he receives no message in the first round, then he sends Alice a $\NO$ message in the second round and never fires. 

It is easy to verify that $\FS$ satisfies~\Spec. The agents never fire when $\go=0$, and if $\go=1$ then Alice fires with probability~1 at time~2, and they both fire at time~2 with probability $0.99\ge 0.95$, as desired. 
Observe, however, that in a run of~$FS$ in which $\go=1$, both of Alice's messages are lost, and Bob's $\NO$ message is delivered to Alice, she fires at time~2 despite being absolutely certain that Bob is not firing.  
The $FS$ protocol in this example shows that in a probabilistic protocol that succeeds with high probability, agents are not always required to act only when they believe that their actions will succeed.%
\footnote{Our example is based on one due to Halpern and Tuttle in~\cite{HalpernTuttle}, in which the same behavior appears. We do not suggest that the protocol~$\FS$ is the most sensible solution to the problem; it is presented here for the sake of analysis. } 
\hfill $\square$
\end{example}

Note that  Alice  
can have three different states of information when she fires in~$\FS$, corresponding to whether she received a $\YES$ message in the second round, she received a $\NO$ message, or received no message from Bob.
In the latter case, Alice knows that Bob's message was lost, but is unsure about what he sent. 
Roughly speaking, in this case Alice ascribes a probability of $.99$ to the event that Bob is firing, since that is the probability that he received at least one of her messages. 
If Alice received a $\YES$ message, then she knows for certain that Bob is firing 
when she fires, while if she received a $\NO$ message, 
she is certain that Bob is {\em not} firing.

%
%

The connection between knowledge and actions that we discussed at the outset applies very broadly and  is not specific to mutual exclusion. Indeed, recent work has shown that it holds generally in distributed systems, for deterministic actions and deterministic goals~\cite{KoP-paper}. 
More precisely, a theorem called the {\em Knowledge of Preconditions Principle} (or \KoP, for short) establishes that if some property~$\cond$ (e.g., ``{\it no agent is in the critical section}'') is a necessary condition for performing an action~$\act$, then an agent must {\em know} that~$\cond$ holds when it performs~$\act$. This is a universal theorem, which applies in all systems, for all actions and conditions, provided that the action is deterministic and  that~the condition must surely hold whenever the action is performed. Explicit use of the KoP has facilitated the design of efficient
protocols for various problems~\cite{Beyond,Unbeatable,Silence}, which improved on the previously best known solutions, sometimes by a significant margin.

This paper seeks to generalize the \KoP\ to probabilistic distributed systems, where both protocols and guarantees may be probabilistic. 
Probabilistic distributed systems are of interest in a wide range of settings. 
Probabilistic protocols are used to facilitate symmetry breaking, load balancing and fault-tolerance \cite{BenOr,FeldmanMicali, Aloha,Upfal,Luby,Grisha,Michal}. Participants in interactive proofs, and more generally in cryptography, typically follow probabilistic protocols \cite{BabaiMoran, GMR}.
Similarly, in many competitive settings such as games, auctions and economic settings, agents follow probabilistic strategies that give rise to probabilistic systems~(see, e.g., \cite{osborne1994course}). 
Often, distributed protocols operate in a context in which the environment (or scheduler) is probabilistic, as in the case of population protocols and many network protocols \cite{angluin2006computation, lindgren2003probabilistic,baras2003probabilistic}. 
Roughly speaking, a distributed system is probabilistic%
\footnote{Technical terms such as a probabilistic distributed system, knowledge, probability and probabilistic beliefs are used loosely in the introduction; 
all relevant notions are formally defined in the later sections.}
 if  the protocol,  the environment, or both, are probabilistic. 
Distributed biological systems such as ant colonies, the brain, and many more, are often best modelled in probabilistic terms
\cite{navlakha2015distributed, colorni1991distributed, afek2011biological, afek2013beeping, feinerman2012memory}. 
Since probabilistic systems are ubiquitous, a good understanding of the interaction between action and information in such systems may provide useful insight into such systems.  

Specifications of deterministic protocols in non-probabilistic systems typically require them to satisfy a set of definite constraints, which must be satisfied in {\em all} executions. 
E.g., all runs of a consensus protocol must satisfy Decision, Validity and Agreement~\cite{lynch1996distributed}. Similarly, all runs of a mutual exclusion protocol must satisfy the exclusion property at all times. For probabilistic systems, correctness may be specified in several ways. In some cases, specifications are definite, and probability is only used to break symmetry or affect the probability or timing of reaching a desired goal. This is the case, for example, for Agreement in Ben-Or's consensus protocol~\cite{BenOr}, or in the mutual exclusion protocols considered in~\cite{pnueli1983extremely}. 
In other cases, however, the protocol is required to succeed with high probability. This is the case, for example, in interactive proofs~\cite{GMR,BabaiMoran}, as well as in several well-known consensus protocols (e.g.,~\cite{rabin1983randomized,FeldmanMicali}) in which disagreement can occur.
Of course, it is similarly possible to relax the correctness of ME protocols, by requiring that the probability of the exclusion property failing should be small. Indeed, in the setting of Example~\ref{ex:1}, no protocol can coordinate attacks and ensure that no agent will ever attack alone. 
We remark that in a related setting agents may be required to act only if they strongly believe that their actions will succeed. Namely, a judge is required to find a defendant guilty only if she believes him to be guilty beyond a reasonable doubt. Taking a probabilistic interpretation, we may take this to mean that a guilty verdict is allowed only if the judge very strongly believes in the defendant's guilt~\cite{whitman2008origins}. (Interestingly, in civil cases in the UK, the requirement for a judgement is that fault be proved on  a ``{\it balance of probabilities}''~\cite{Reece1996losses}. This means, roughly speaking,  that one scenario is believed to be more likely than its converse.)

Given the probabilistic nature of events in a probabilistic system, in addition to knowledge that certain events hold in an execution, an agent may have probabilistic beliefs about relevant facts. Very roughly speaking (and informally for now), let us denote by $\pBel_i(\varphi)$ agent~$i$'s degree of belief that a given fact~$\varphi$ holds. In the protocol~$\FS$, for example, if Alice receives a~$\YES$ message from Bob at time~2,  then 
$\pBel_A(\fireB)=1$ (where $\pBel_A$ stands for Alice's belief and~$\fireB$ stands for ``{\it Bob is firing}''), while  $\pBel_A(\fireB)=0$ if she receives the message~$\NO$. 
In case Bob's message is lost, and so she receives neither, Alice is unsure whether Bob is firing, but her degree of belief in it is high (indeed, $\pBel_A(\fireB)=0.99$ in this case). 

 Our investigation will be most closely related to probabilistic guarantees in which protocols are required to succeed with high probability. Motivated by, e.g., the relaxed firing squad problem and relaxing mutual exclusion,  we are interested in guarantees 
 that a certain condition (or fact) $\cond$ should be true with high probability, when (or given that) a particular action~$\act$ is performed. 
 We will call such a requirement,  which we  informally denote by \mbox{$\mu(\varphi\gvn\act)\ge p$}, a \defemph{probabilistic constraint} for a protocol. 
 In Example~\ref{ex:1} the probabilistic constraint can be expressed by
 $\mu(\both\gvn\fireA)\ge 0.95$, where $\both$ is the fact that both agents are currently firing, and $\fireA$ is the fact that Alice is firing. 
 
 We think of the value~$p$ in a probabilistic constraint of the form $\mu(\varphi\gvn\act)\ge p$ as the {\em desired threshold} probability that the solution should achieve. 
It is natural to think of an agent's beliefs as {\em meeting the threshold} of the probabilistic constraint $\mu(\varphi\gvn\act)\ge p$ at a point where~$i$ performs~$\act$ in a given execution if $\pBel_i(\varphi)\ge p$ at that point.  We are interested in studying the interaction between probabilistic constraints on an agent's actions, and her probabilistic beliefs when acting.
In Example~\ref{ex:1}, the~$\FS$ protocol satisfies the probabilistic constraint, even though the threshold is not always met when Alice fires. 
We notice, however, that this happens infrequently---Alice fires without her beliefs meeting the threshold only with a probability of $0.009=0.1\cdot0.1\cdot0.9$. In a measure $0.991$ of the runs in which Alice fires, the threshold is met when she fires. We remark that $0.991\ge 0.95$. Is this a coincidence, or can we prove that the threshold must be met whp? 
In addition to analyzing the necessary conditions on beliefs for satisfying a probabilistic constraint, we are interested in sufficient conditions. It is natural to conjecture that always meeting the threshold is a sufficient condition for satisfying a probabilistic constrain. Is that indeed the case?

Our investigation will be made with respect to a class of probabilistic systems that satisfy several simplifying assumptions. Nevertheless, the answers that it provides will offer new insights into the connection between actions and probabilistic beliefs. In particular, we will consider finite purely probabilistic systems (pps for short). 

$~$\\[-1.8ex]
The main contributions of this paper are:\\[-2.2ex]
\begin{itemize}[leftmargin=.15in]
\item We initiate a systematic study of the connection between actions and probabilistic beliefs for a wide range of probabilistic protocols, and for general conditions~$\varphi$. In particular, we consider beliefs when agents perform {\em mixed} actions, in which, e.g., a concrete action~$\act$ is performed only with a certain probability. Modeling, formulating, and proving these connections
is a subtle matter.
\item We first consider sufficient conditions on beliefs for ensuring that a probabilistic constraint is satisfied. Perhaps unexpectedly, always meeting the threshold is not, in general, a sufficient condition. It may fail to be sufficient when agents perform mixed actions (i.e., the choice of action at a point is a probabilistic function of the local state). 
But all is not lost. We identify an independence condition of the condition~$\varphi$ from the action~$\act$, under which it is shown to be sufficient. Moreover, the independence condition appears to hold in most cases of practical interest.  
\item It is shown that the threshold need not be met whp for a probabilistic constraint to be satisfied. In particular, we prove that there is \defemph{no positive lower bound} $\boldsymbol{\vareps>0}$  on the measure of runs in which the threshold must be met when~$\act$ is performed in order for a probabilistic constraint to be satisfied. 
\item 
We show that the {\em expected value} of the degree of belief plays a central role in establishing the probabilistic constraints. 
Our main theorem is that if~$\varphi$ and~$\act$ satisfy the independence property mentioned above, then the {\em expected  value} of~$\pBel_i(\varphi)@\act$ in the system is equal to $\mu(\varphi\gvn\act)$. 
\item As a corollary of this theorem, we can prove that in order to satisfy a probabilistic constraint, an agent must, with high probability,  strongly believe that the condition holds when it acts. 
More formally, suppose that $\mu(\cond\gvn\act)=1-\vareps^2$ for some~$\vareps\in[0,1]$. 
Then  $\pBel_i(\varphi)@\act \ge  1-\vareps$ in probabilistic measure at least $1-\vareps$ of the runs in which~$i$ performs~$\act$. 
Of course, if $1>\vareps>0$ then $1-\vareps$ is smaller than the constraint's threshold
of $1-\vareps^2$. However, if~$\vareps$ is small, this means that~$i$ must \defemph{probably approximately know} that~$\cond$ holds when it performs~$\act$.
\end{itemize}

The 1980's saw the emergence of formal models of knowledge and beliefs in distributed systems. A thorough presentation elucidating a variety of subtle aspects involved in modeling and reasoning about probabilistic beliefs appears in~\cite{HalpernUncertainty}. While Fagin and Halpern~\cite{fagin1994reasoning} presented a general model in which agents' probabilistic beliefs can be expressed, our presentation is most closely related to \cite{FischerZuck,HalpernTuttle}. Following~\cite{FischerZuck}, we model a probabilistic system in terms of a synchronous execution tree whose edges are labeled by transition probabilities. 
In \cite{FischerZuck}, Fischer and Zuck state that if a deterministic protocol for coordinated attack guarantees that an attack is coordinated with probability~$p$, then the ``average'' belief of~$A$ 
in the fact that~$B$ is attacking, when it attacks, is at least~$p$.
A closer reading of \cite{FischerZuck} reveals that this property of the average belief is precisely the probabilistic constraint that A must, with probability at least~$p$,  believe that both are attacking when~$A$ 
attacks.  
We take the investigation one step further, and characterize what $A$'s beliefs need to satisfy in order to ensure that such a constraint is satisfied. 
Moreover, while they consider two concrete examples and deterministic protocols, 
we investigate arbitrary probabilistic constraints, in a setting that allows for general probabilistic protocols. 
Halpern and Tuttle consider several different notions of probabilistic beliefs, and show that they correspond to different modelling assumptions. They relate probabilistic beliefs to notions of safe bets, and discuss how coordinated attack is related to different notions of belief and to probabilistic common belief. Our notion of probabilistic beliefs is what~\cite{HalpernTuttle} refer to as $\mathcal{P}^{\mathrm{post}}$, or the agent's posterior beliefs obtained by conditioning on her local state.

Reasoning about agents' probabilistic beliefs is well established in game theory~\cite{aumann1976agreeing,osborne1994course,SametMonderer1989}. 
In this literature, agents are typically assumed to possess a {\em common prior}, which is a central property of the purely probabilistic systems that we consider. Their models are normally based on a fixed universe of states of the world, in which actions are not explicitly modeled. Fagin and Halpern~\cite{fagin1994reasoning} as well as Monderer and Samet~\cite{SametMonderer1989} present a novel notion of probabilistic common beliefs and discuss its applicability. 


This paper is organized as follows. The following two sections present our model of probabilistic systems, and the notion of subjective probabilistic beliefs. 
Analyzing probabilistic constraints of the form described above,
\Cref{sec:sec4label} shows that 
under a certain independence assumption,
meeting the threshold of a probabilistic constraint (i.e., holding a strong belief) is a sufficient condition for satisfying it.
\Cref{sec:sec7label} shows that 
the threshold must at least sometimes be met, but there is no lower bound on the measure of runs in which it must  be met. 
\Cref{sec:sec5label} presents our main result, proving that in order to satisfy a probabilistic constraint, the agent should, in expectation, hold a strong belief.
Finally, \Cref{sec:sec6label} shows that in order to satisfy a probabilistic constraint an agent should hold a strong belief with high probability. 
The Appendix contains proofs of all formal statements in the paper. 


\section{Model and Preliminary Definitions}
\label{sec:model}
Reasoning about knowledge, beliefs, and probability in distributed systems can be rather subtle, and it has received extensive treatment in the literature over the last four decades. The foundations of our modeling of  distributed systems are based on the interpreted systems framework of \cite{FHMV}, and the modeling of probability and probabilistic beliefs in distributed systems is based on \cite{FischerZuck,HalpernUncertainty,HalpernTuttle}. 
An important issue that arises when considering probabilistic beliefs (as discussed extensively in~\cite{FischerZuck,HalpernTuttle}) has to do with the interaction between  nondeterministic choices and probabilistic beliefs. 
For example, consider a model that differs from that of Example~\ref{ex:1} only in that the value of Alice's local variable~$\go$ is set nondeterministically, rather than probabilistically. In a run of the protocol~$\FS$ in that model in which  Bob  does not receive any message at time~1, 
what can we say about Bob's beliefs at  time~2 regarding whether Alice is firing? Roughly speaking, $\fireA$ is not a measurable event for Bob at that point, because we have made no probabilistic assumptions about whether~$A$'s flag is initially in the~$\go=1$ or the~$\go=0$ state.  
We can think of the state of the flag as being a nondeterministic choice made by the scheduler\footnote{The scheduler can be thought of as being ``Nature" or ``the environment;'' we use the terms interchangeably.} before
time~0. Similar issues regarding measurability arise with other nondeterministic actions by the scheduler or by the agents. In a probabilistic protocol for consensus, for example, it is typically assumed that the scheduler can freely determine who the faulty agents are, and how they act. 
As pointed out by Pnueli~\cite{pnueli1983extremely} and discussed by~\cite{FischerZuck,HalpernUncertainty,HalpernTuttle}, the way to formally handle reasoning about probabilities in the presence of nondeterminism is to 
consider as fixed the set of all nondeterministic choices in an execution. 
Halpern and Tuttle
in \cite{HalpernTuttle} consider fixing such a set as determining a particular ``adversary.''
In particular, an adversary can determine a unique initial global state or, more generally, a distribution over initial states. 
Once we fix the adversary, all choices, whether those by the scheduler or those by the agents, are purely probabilistic. 
We will study the relation between actions and probabilistic beliefs when a protocol is executed in the context of a fixed adversary. In this case, the set of runs of the protocol can be modeled as a tree; 
we proceed as follows.

%

\subsection{Probabilistic Systems}
\label{sec:systems}

We assume a set $\Proc=\{1,2,\ldots,n\}$ of~$n$ agents, and a scheduler, denoted by~$e$, which we call the {\em environment}. 
A {\em global state} is a tuple of the form $g=(\lEnv,\ell_1,\ell_2,\ldots,\ell_n)$ associating a {\em local state} $\li$ with every agent~$i$ and a state~$\lEnv$ with the environment.
We model a {\em purely probabilistic system} (pps) 
by a finite labelled directed tree of the form $T=(V,E,\pi)$, where $\pi:E\to(0,1]$ assigns a probability to each edge of~$T$. In particular,  
$\sum\limits_{w\in V}\pi(v,w)=1$ holds for every internal node $v$ of~$T$. 
All nodes of~$T$ other than its root correspond to global states. 
 The root is denoted by~$\rootv$, and its sole purpose is to define a distribution over its children, which represent {\em initial} global states. Every path from one of~ the root's children to a leaf is considered a {\em run} of~$T$, and we denote by~$\RT$ the set of runs of~$T$.
A run is thus a finite sequence of global states.
We denote the initial global state of a run~\mbox{$r\in\RT$} (which is a child of the root~$\rootv$ in~$T$) by $r(0)$ and its~$k+1^\mathrm{st}$ global state by~$r(k)$. 
Agent~$i$'s local state at the global state $r(t)$ is denoted by $\rit$.
We shall restrict attention to synchronous systems, meaning that the agents have access to the current time. 
Formally, we assume that every local state of agent~$i$ contains a variable~$\timei$, and whenever $\li=\rit$ in~$T$  
the value of~$\timei$ in~$\li$ equals~$t$.  Intuitively, this guarantees that every agent will always know what the current time is. (We restrict attention to synchronous systems, since modeling probabilistic beliefs in asynchronous systems is nontrivial, as discussed  in~\cite{HalpernUncertainty,HalpernTuttle}.) 
Formally, a pps~$T$  induces a probability space $\XT=(\RT,2^\RT,\muT)$ over the runs of~$T$. %
(This is commonly called a {\em prior} probability distribution over the set of runs.)
The probability distribution~$\mu=\muT$ is defined as follows. For a  run~$r=v_0,v_1,\ldots,v_k$, we write $\mu(r)$ instead of $\mu(\{r\})$, and define 
$\mu(r)=\pi(\rootv,v_0)\cdot\pi(v_0,v_1)\cdots\pi(v_{k-1},v_k)$. Thus, the probability of a run is the product of the probability of its initial global state and the transition probabilities along its edges. Based on our assumptions regarding~$T$, it is easy to verify that~$\XT$ is a probability space. 
Since~$\RT$ is finite and every run of~$\RT$ is measurable, every subset $Q\subseteq\RT$ is measurable, and 
$\mu(Q)=\sum\limits_{r\in Q}\mu(r)$. 

\subsection{Relating Protocols to Probabilistic Systems}\label{sec:Relating Protocols to Probabilistic Systems}

Let ~\iInNEnv\ denote an agent or the environment.  Denote by~$\Li$ the set of~$i$'s  local states, and by~$\Act_i$ the set of local actions it performs in a given protocol of interest. For simplicity, we assume that the sets~$\Act_i$ are disjoint.
A \textbf{(probabilistic) protocol} for~$i$ is a function $\Proti:\Li\to\Delta(\Act_i)$ mapping each local state $\li\in\Li$ to a distribution over~$\Act_i$. 
This distribution determines the probabilities by which~$i$'s action at~$\li$ is chosen. 
We assume that $\Proti(\li)$ assigns positive probability to a finite subset of $\Act_i$ for  every $\li\in \Li$.
When $\Proti(\li)$ assigns positive probability to more than one action, we say that the agent is performing a \defemph{mixed action step}, using the language of game theorists~\cite{osborne1994course}. The probabilistic choice in this case is made based on the local state~$\li$, and when the agent decides on the mixed step she does not know which of the actions in its support will actually be performed.

A {\em joint} protocol is a tuple $P=(P_e,P_1,\ldots,P_n)$. 
We will restrict attention to systems~$T$ in which, at every non-final point (i.e., a point that does not correspond to a leaf in the tree), the environment and each of the agents perform an action. 
Every tuple of actions performed at a global state~$g$ determines a unique successor state~$g'$ 
as well as the probability~$\pi(g,g')$ of transition from~$g$ to~$g'$.

Given a probability distribution over the finite set of initial global states, if the environment and all agents follow probabilistic protocols that terminate in bounded time as above, then the set of runs of the system can be modeled by a pps~$T$. 
Indeed, since we assume that the support of $\Proti(\li)$ is finite in all cases, the number of successors of a global state~$g$ in runs of such a joint protocol~$P$ 
is finite.
In the setting of Example~\ref{ex:1}, since we are given a fixed probability of~0.5  that $\go=0$ in the initial global state, and a probability of 0.5 that $\go=1$,  the set of runs of the protocol~$\FS$ can be represented by a pps.
(If the value of~$\go$ were set nondeterministically, then the initial global state with $\go=0$ would define a pps, and the one with $\go=1$ would define another, seperate, pps; see \cite{HalpernUncertainty,HalpernTuttle}.)

In the sequel, we will need to keep track of what actions are performed in any given state by the various agents and by the environment. 
To this end, we will assume w.l.o.g.\ that at every global state  the environment's local state~$\ell_e$ contains a ``history'' component~$h$ that is a list of all actions performed so far, when each action was performed, and by which agent. 


\subsection{Facts in Probabilistic Systems}
Due to space limitations, we will not present a formal logic for reasoning about uncertainty in distributed systems (for this, the reader should consult~\cite{HalpernUncertainty}). Rather, we will cover just enough of the definitions to justify our investigation.  

We are interested in reasoning about conditions and facts such as whether agents perform particular actions, what agents' initial values were, etc. While some of these are properties of the run, others (e.g., ``{\it the critical section is empty}'') are transient, in the sense that they refer to the state of affairs at the current time, and their truth value can change from one time to another. 
%
Therefore,  we will consider the truth of 
facts 
at {\em points} $(r,t)$, which refer to time~$t$ in a run~$r$. 
We denote by  $\Pts(T)=\{(r,t): r\in\RT~\mbox{and}~r(t)~\mbox{is a node of}~ T\}$ the set of points of a pps~$T$.

A fact (or event) over a pps~$T$ is identified with a subset of $\Pts(T)$ which, intuitively, is the set of points at which the fact is true. 
We write $(T,r,t)\sat \varphi$ to denote that~$\varphi$ is true at the point $(r,t)$ of the system~$T$. 
For example,  we will later use the fact~$\acti$ stating that~$i$ is currently performing~$\act$.
Formally, we define 
$(T,r,t)\sat\acti$ to hold iff the history component $h$ 
 of actions in $r_e(t+1)$ records that $i$ performed~$\act$ at time~$t$.
In a similar fashion, we write $(T,r,t)\models \neg\varphi$ to denote that a fact~$\varphi$ is not true at the point $(r,t)$ of the system~$T$.

In some cases, we are interested in facts~$\psi$ that are properties of the run, such as ``{\em all agents decide on the same value}.'' 
Formally, we say that a fact~$\psi$ is a {\em fact about runs} in the system~$T$ if for all $r\in\RT$ and all 
times $t,t'\ge 0$ it is the case that $(T,r,t)\sat \psi$ iff $(T,r,t')\sat \psi$.
For facts~$\psi$ about runs, we write $(T,r)\sat\psi$ to state that~$\psi$ is true in the run~$r$ of the system~$T$.

Intuitively, for a transient fact~$\varphi$, the fact ``{\it $\varphi$ holds at some point in the current run}'' is a fact about runs. 
On several occasions, we will be interested in run-based facts of this type. 
In particular, for an action~$\act\in\Acti$ and for a local state~$\li\in\Li$, we will use $\racti$ and~$\rli$ to denote facts about runs defined as follows: 
\begin{align*} 
(T,r)\sat\racti\quad &\mbox{iff}\quad (T,r,t)\sat\acti~\mbox{holds for some time~$t$} ~. \\[.5em]
(T,r)\sat\rli\quad &\mbox{iff}\quad \rit=\li~\mbox{for some time }t ~. 
\end{align*}

%


\section{Probabilistic Beliefs}\label{sec:Beliefs in Sync sys}

An event in the probability space~$\XT$ consists of a set $Q\subseteq\RT$ of runs. 
Since we can view a fact~$\theta$ about runs as corresponding to the set of runs that satisfy~$\theta$, we will often abuse notation slightly and treat facts about runs as representing events. 
We will thus be able to consider the probabilities of such facts, and as well as to condition events on facts about runs.

As is common in the analysis of distributed systems using knowledge theory, we identify the local information available to~$i$ at a given point with its local state there (see~\cite{FHMV}). 
In the current setting, in order to capture probabilistic beliefs, we associate with every local state $\li\in \Li$ in~$T$ the probability space 
\mbox{${\mathcal X}^{\li}=\big(\RT, 2^{\RT},\muT(\cdot\given\rli)\big)$}.
Considering the probability measure~$\muT$ induced by~$T$ as a prior probability measure on sets of runs, the assignment $\muT(\cdot\given\rli)$ captures agent~$i$'s subjective posterior probability. 
Recall that, by definition, $\pi$ assigns positive probabilities to all transitions in a pps~$T$. Consequently, $\muT(r)>0$ for every run $r\in\RT$, and hence \mbox{$\muT(\rli)>0$} for every local state~$\li$ that appears in~$T$.  
It follows that $\muT(Q\given\rli)$ is well-defined for every set of runs~$Q\subseteq\RT$.


Our investigation of  beliefs in the context of probabilistic constraints will focus on whether (possibly transient) facts of interest hold {\em when} an agent acts, or when the agent is in a given local state. Since the current time is always a component of an agent's local state in a pps, a given local state~$\li$ can appear at most once in any particular run~$r$. This facilitates the following notation.  
We write $\varphi@\li$, for a state $\li\in\Li$, to state that~$\varphi$ holds when~$i$ is in local state~$\li$ in the current run. Formally, we define 
$(T,r)\models \varphi@\li$ iff both $(T,r)\models\rli$ (the local state occurs in~$r$), and $(T,r,t)\models \varphi$ holds for the point $(r,t)$ at which $\rit=\li$.

We can now define an agent~$i$'s degree of probabilistic belief in a fact~$\varphi$ at a given point as follows: 
\begin{definition}
The value of $\pBel_i(\varphi)$ at  $(r,t)\in\Pts(T)$ is defined to be $\muT(\varphi@\li\given\rli)$, where~${\li=\rit}$. 
\end{definition}
 
The value of $\pBel_i(\varphi)$ represents~$i$'s current degree of belief that~$\varphi$ is true. It depends on~$i$'s local state, and changes over time as the local state changes. 

\subsection{Proper Actions}
Recall that probabilistic constraints impose restrictions on the conditions under which an action can be performed. In some cases, these are conditions are facts about the the run (e.g., {\em ``all processes decide~0 in the current run''} or {\em ``all initial values were~1''}), and in other cases they may be transient facts such as {\em ``the critical section is currently empty''}. Transient facts do not, in general correspond to measurable events in our probability space~$\XT$. 
%
%
To overcome this difficulty we will restrict our attention to actions that are performed at most once in any given execution of the protocol. We proceed as follows:

We say that~$\boldsymbol{\act}$  \textbf{\em is a proper action for~$\boldsymbol{i}$ in~$\boldsymbol{T}$} if~$i$ performs~$\act$ at least once in~$T$ and, for every run \mbox{$r\in \RT$}, agent~$i$ performs~$\act$ at most once in~$r$. 
For a proper action~$\act$, the set of runs in which~$\act$ is performed is well-defined. 
Moreover, in such a run, the time, as well as the local state, at which~$\act$ is performed are unique. Technically, restricting actions to be proper will enable us to partition the set of runs in which an action~$\act$ is performed according to the local state at which~$i$ performs~$\act$. 

Our analysis will focus on an agent's beliefs {\it when it performs} a proper action. 
Restricting attention to proper actions does not impose a significant loss of generality. Either tagging an action with its occurrence index (e.g., ``{\em the third time~$i$ performs~$\act$}'') or timestamping  actions with the time at which they are performed (``{\em the action~$\act$ performed by~$i$ at time~$t$}''), 
can be used to convert any given action into a proper one. 

For a proper action~$\act\in\Acti$, 
we take $\varphi@\act$ to be a fact stating that~$\varphi$ holds in the current run when~$i$ performs~$\act$. Since~$\act$ is proper, this is a fact about runs.
Formally, we define
    $(T,r)\models \varphi@\act$ to hold iff both $(T,r)\models\racti$ (i.e., $\act$ is performed in the current run~$r$) and  $(T,r,t)\models \varphi$ holds for the (single) point $(r,t)$ of~$r$ at which $(T,r,t)\models \acti$.

Since we are interested in~$i$'s beliefs when she performs an action, we will similarly use $\boldsymbol{\pBel_i(\varphi)@\act}$ to 
refer to~$i$'s degree of belief in~$\varphi$ when it performs~$\act$. Formally, we define  $\boldsymbol{\big(\pBel_i(\varphi)@\act\big)[r]}$ to be the value of~$\pBel_i(\varphi)$ at the point $(r,t)$ at which $(T,r,t)\models \acti$.
By convention, if~$i$ does not perform~$\act$ in~$r$, then $\big(\pBel_i(\varphi)@\act\big)[r]=0$ for every fact~$\varphi$. 

\begin{definition}
\label{def:prob-const}
A \defemph{probabilistic constraint} on an action~$\act$ in a pps~$T$ is a statement of the form \[\boldsymbol{\muT(\varphi@\act \Bigiven \ract) \ge p}~~.\]

\end{definition}

For a fact~$\psi$ about runs, 
the form of a probabilistic constraint becomes much simpler. In this case, $(T,r,t)\models \psi@\act$ iff both \mbox{$(T,r)\models \psi$} and $(T,r)\models \ract$ hold. Since~$\muT(r \Bigiven \ract)=0$ for every run~$r$ at which~$\act$ is not performed by~$i$, the constraint becomes simply $\boldsymbol{\muT(\psi \Bigiven \ract)\ge p}$, 
since $\muT(\psi@\act \Bigiven \ract)=\muT(\psi \Bigiven \ract)$.

We are now ready to start our formal investigation.



\section{The Sufficiency of Meeting the Threshold}
\label{sec:sec4label}


%

Intuitively, acting only under strong beliefs should suffice for guaranteeing probabilistic constraints. 
Namely, we would expect that 
for every proper action~$\act$ of~$i$ and fact~$\varphi$, 
if 
$\pBel_i(\varphi)\ge p$ 
holds whenever %
~$i$ performs~$\act$ in~$T$, then 
    ~$\muT(\varphi@\act \gvn \racti)\ge p$ (i.e., if~$i$ performs~$\act$ only when its belief in~$\varphi$ meets the threshold of a given constraint, then the constraint will be satisfied). 
    This is indeed true in many cases of interest. Perhaps somewhat unexpectedly, it is not true in general. 
     For an example in which it fails,  consider the system~$T$ depicted in \Cref{fig:dependence} in which there is a single agent, called~$i$, and a single initial global state~$g_0$. At time~$0$ the agent~$i$ performs either~$\act$ or~$\act'$, each with probability~$\frac{1}{2}$. The resulting pps~$T$ contains two runs, $r$ in which~$i$ performs $\act$, and~$r'$ in which it performs~$\act'\ne\act$. Let the fact of interest be $\psi=\neg\acti$. 
It is easy to check that $\muT(\psi@\act\Bigiven\,\racti)=0$, since by definition, $\act$ is performed precisely whenever~$\psi$ is false.  
As far as~$i$'s beliefs are concerned, $\pBel_i(\psi)=\frac{1}{2}$ when~$i$ performs~$\act$, since~$i$'s local state at the initial global state~$g_0$ guarantees with probability~$\frac{1}{2}$ that~$\act$ will not be performed.
We thus have that $\pBel_i(\psi)\ge \frac{1}{2}$ whenever~$i$ performs~$\act$ in~$T$, while 
$\muT(\psi@\act \gvn \racti)=0<\frac{1}{2}$.

 \begin{figure}[h]
\centering
\includegraphics[width=3in]{
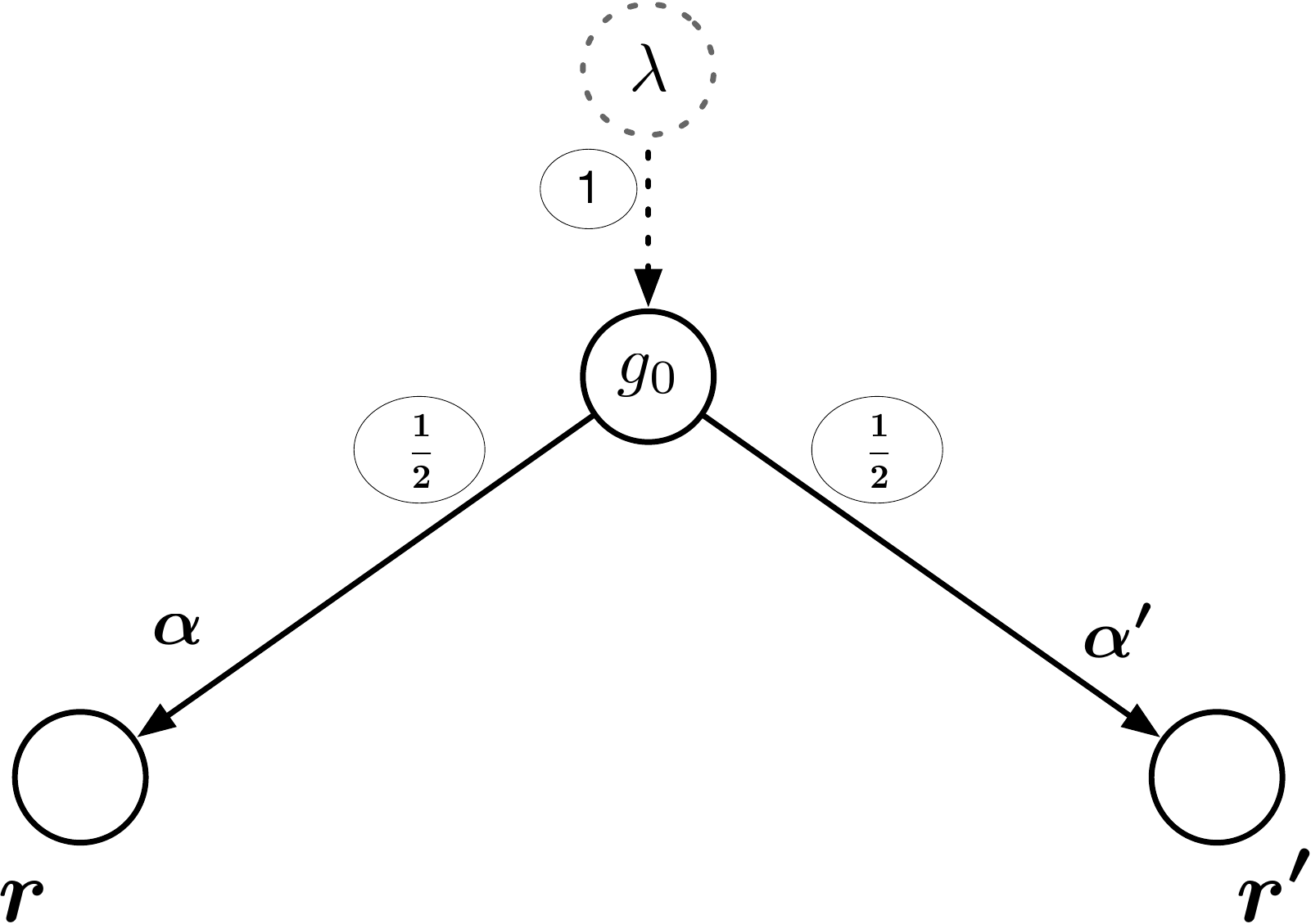}
\caption{An example of a pps~$T$ 
.}
\label{fig:dependence}
\end{figure}

\vspace{5mm}
In this example, the condition~$\psi=\neg\acti$ of interest depends strongly on whether the action~$\act$ is performed. This is unlikely to be the case in typical probabilistic constraints. We remark, however, that the claim can be shown to fail in more natural scenarios, such as when the action~$\act$ consists of sending a particular message and~$\psi$ depends on whether its recipient  acts in a particular way in a future round. 

The problem arises from the dependence between $\varphi$ and~$\acti$. We now present  an independence assumption that holds in many cases of interest, under which the desired property holds. 

\begin{definition}\label{def:local-state indep}
Let~$\act$ be a proper action for agent~$i$ in~$T$. We say that~$\varphi$ is \textbf{\em local-state independent} of~$\act$ in~$T$ if, for all $\li\in\Li$ it is the case that
$$\muT(\varphi@\li\Bigiven\,\rli) \cdot \muT(\act@\li\Bigiven\rli)~~=~~\muT\big([\varphi\wedge \act]@\li \Bigiven\rli \big)~.$$
\end{definition}

Intuitively, local-state independence implies that the probability that~$\varphi$ will hold when agent~$i$ performs the action~$\act$ is  independent of the local state at which~$\act$ is performed. 
We can now show:

\begin{restatable}{theorem}{recallBeliefImpliesProb}\label{thm:belief implies prob}
Let $\act$ be a proper action for agent~$i$ in a pps~$T$, and let~$\varphi$ be a fact that is local-state independent of~$\act$ in~$T$.
If $\pBel_i(\varphi)\ge p$ at every point of~$T$ at which~$i$ performs~$\act$, then 
    ~$\muT(\varphi@\act \gvn \racti)\ge p$.
\end{restatable}

\begin{proof}[Proof sketch]$\!\!\!$%
 ~We partition the event $\Ract$ consisting of the runs of~$\RT$ in which~$\act$ is performed according to~$i$'s local state~$\ell_i$ when it performs~$\act$. Using~$\varphi$'s local-state independence of~$\act$ in~$T$ we show that, in every cell of the partition, $\varphi$ holds with probability at least~$p$ when~$\act$ is performed. The claim then follows by the law of total probability.
\end{proof}

\Cref{thm:belief implies prob} can be viewed as following from the Jeffrey conditionalization theorem in probability~\cite{Jeffrey}. Roughly speaking, Jeffrey conditionalization 
relates the prior  probability of an event when an experiment is performed, to its posterior probabilities given the possible outcomes of the experiment. 
As discussed in \Cref{sec:Beliefs in Sync sys}, and agent's probabilistic beliefs coincide with the posterior probability, obtained by conditioning $\muT$ on the realized local state. 
 Jeffrey conditionalization is also the basis of our main theorem, and we discuss it in slightly greater detail in \Cref{sec:Jeff}. 
As a statement about prior probabilities and probabilistic beliefs, \Cref{thm:belief implies prob} generalizes a result that Samet and Monderer \cite{SametMonderer1989} proved in a simpler setting. 
They considered a model that corresponds to a static system in which performing actions is not explicitly modeled. In our formalism, this would correspond to a ``flat'' pps~$T$ consisting only of a root node and its children (corresponding to initial states). They showed that in such a system, if an agent's expected degree of (posterior) belief regarding a fact~$\varphi$ is greater or equal to a value~$p$, then the objective (prior) probability of~$\varphi$ is, in itself, at least~$p$.

While the local-state independence property of \Cref{def:local-state indep} is needed for our proof of \Cref{thm:belief implies prob}, the theorem still applies in many (perhaps most) cases of interest.  
One case in which the problem does not arise is if~$\act$ never participates in a mixed action step. 
More formally, $\act\in\Act_i$ is called a \defemph{deterministic action} for~$i$ in a pps~$T$ if $\acti$ is a deterministic function of~$i$'s local state in~$T$.  I.e., if $\rit=r'_i(t)$ for two points $(r,t),(r',t)\in\Pts(T)$, then~$i$ either performs~$\act$ at both points, or does so at neither of them. 
Even for actions that participate in mixed action steps, local-state independence is guaranteed in many typical cases. 
We say that two runs $r,r'\in\RT$ \defemph{agree up to time~$\boldsymbol{t}$}
~if they share the same prefix up to and including time~$t$. (I.e., if they extend the same time~$t$ node in~$T$.) 
We call $\varphi$ a \defemph{past-based} fact in $T$ if, for all pairs~$r,r'\in \RT$ of runs and times~$t\ge 0$, if~$r$ and~$r'$ agree up to time~$t$, then $(T,r,t)\models \varphi$ exactly if $(T,r',t)\models \varphi$.
Many reasonable conditions, including any fact about the current state of the system such as ``{\it $A$ is attacking}'', or ``{\it the critical section is empty},'' are past based. 
Sufficient conditions for local-state independence are given by: 

\begin{restatable}{lemma}{recallIndependence}\label{lem:independence}
Let~$\act\in\Acti$ be a proper action in a pps~$T$, and let~$\varphi$ be a fact over~$T$. If (a) $\act$ is deterministic in~$T$, or 
(b) $\varphi$ is past-based in~$T$, then 
$\varphi$ is local-state  independent of~$\act$ in~$T$. 
\end{restatable}

\section{On the Necessity of Meeting the Threshold}\label{sec:sec7label}
Recall from protocol~$\FS$ that it is possible to satisfy the probabilistic specification of the relaxed firing squad problem, while allowing actions to be performed even in cases in which the probabilistic property of interest is not strongly believed (indeed, the example shows that it might not be believed {\em at all} in some, rare, cases). 
We can now ask, if we are required to satisfy a given probabilistic constraint of the form~$\muT(\varphi@\act \Bigiven \ract)\ge p$, what can be said about the probability that the agent's belief regarding~$\varphi$ when acting 
meets the threshold~$p$ of the probabilistic constraint? 
%
Indeed, we can show that there must be at least some cases in which belief in the property~$\varphi$ is as high as~$p$. 
More formally: 


\begin{restatable}{lemma}{recallSometimes}\label{lem:sometimes}
Let $\act$ be a proper action for agent~$i$ in a pps~$T$, and let a fact~$\varphi$ be local-state independent of~$\act$ in~$T$.
If \mbox{$\muT(\varphi@\act\Bigiven\,\racti)\ge p$}, then there must be at least one point $(r,t)$ of~$~T$ at which $\act$ is performed and $(T,r,t)\sat \,\,\pBel_i(\varphi)\ge p$.
\end{restatable}

Although there must be some points at which $\pBel_i(\varphi)\ge p$ holds when~$i$ acts,
there is no lower bound on the measure of runs in which the agent must hold such a strong belief when performing~$\act$.
It follows that a probabilistic requirement can be met 
even when the agent's degree of belief in the condition~$\varphi$ rarely meets the threshold set by the probabilistic constraint. 
More formally: 

\begin{restatable}{theorem}{recallNoBound}\label{thm:no-bound}
For every $\vareps>0$ and every $0<p<1$, there exists a pps~$\hat{T}$, a proper action~$\act$ for~$i$ in~$\hat{T}$ and a fact~$\varphi$ which is local-state independent of $\alpha$ 
in~$\hat{T}$  such that~$\mu_{\hat{T}}(\varphi@\act\Bigiven\,\racti)\,\ge p$, but $\mu_{\hat{T}}\big(\pBel_i(\varphi)@\act\ge p \Bigiven\, \racti \big)~\le~ \vareps$. 
\end{restatable}

The proof of this theorem is obtained by presenting a construction of a pps~$\hat{T}$ in which~$\act$ is performed with beliefs that are slightly below the threshold~$p$ in most cases (i.e., whp), and on a small measure of the runs it is performed when that agent's belief ascribes probability~1 to the condition~$\varphi$. 

%
\begin{proof}
\begin{figure}[t]
\centering
\includegraphics[width=6in]{
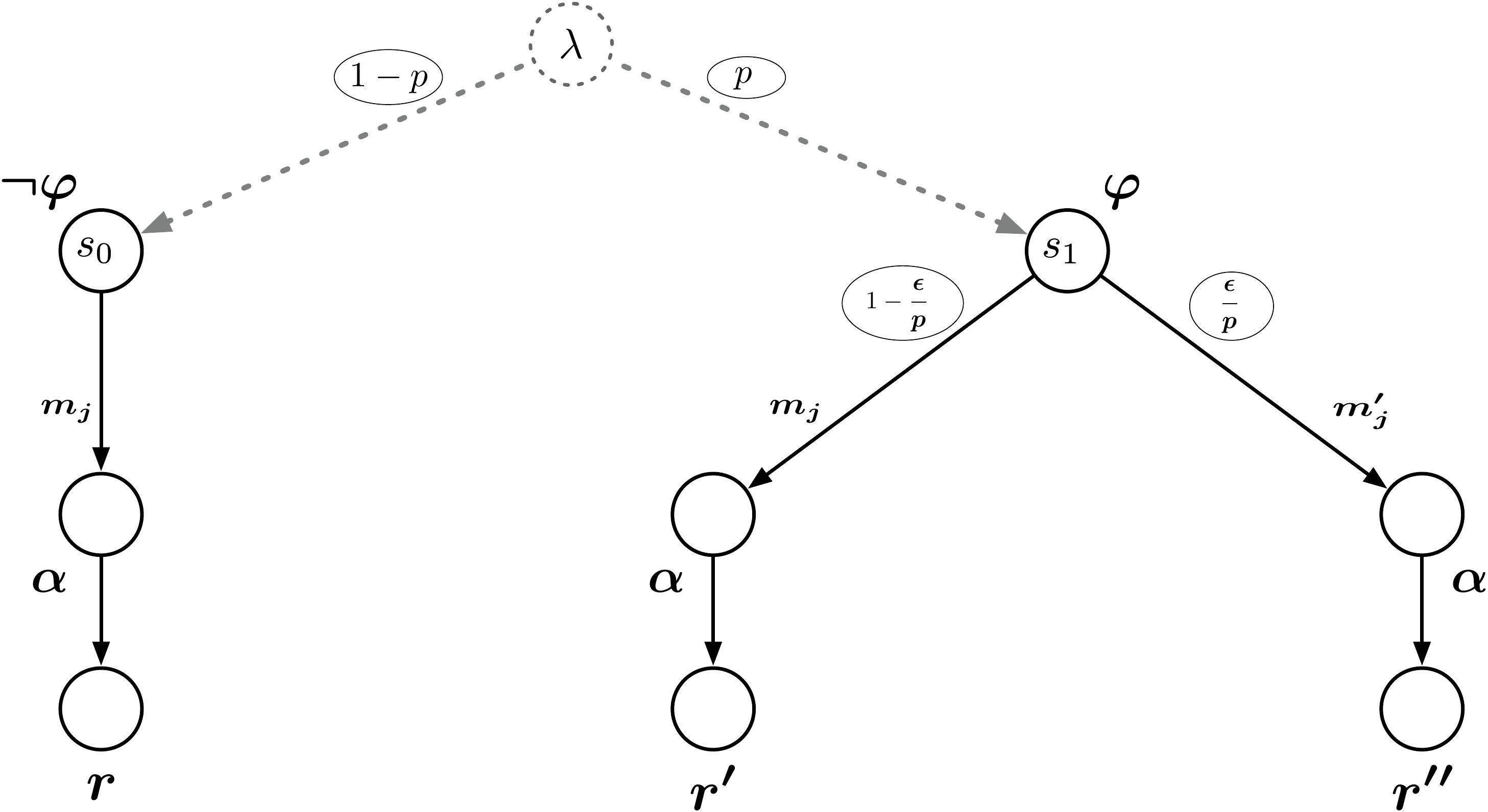}
\caption{The pps $\hat{T}$ described in the proof of \Cref{thm:no-bound}}
\label{fig:Doc1}
\end{figure}

It suffices to prove the claim under the assumption that $0<\vareps<p<1$. So fix such~$p$ and~$\vareps$, and let~$\hat{T}$ be the pps corresponding to the following system, depicted in \Cref{fig:Doc1}:
There are two agents, called~$i$ and~$j$. Agent~$j$'s local state contains a binary value called~`$\bit$' that does not change over time. 
There are two initial global states $s_0$ and~$s_1$, with  $\bit=0$ in the state~$s_0$ and  $\bit=1$ in~$s_1$. Assume that the initial state~$s_1$ is chosen with probability~$p$, while~$s_0$ is chosen with probability~$1-p$. 
In the first round, agent~$j$ acts as follows.  If $\bit=0$ then~$j$ sends the message~$m_j$ to~$i$. If $\bit=1$ then ~$j$'s move is probabilistic; it sends~$i$ the message~$m_j$ with probability $1-\frac{\vareps}{p}$, and it sends~$i$ a message $m_j'\ne m_j$ with probability~$\frac{\vareps}{p}$.  Agent~$i$ receives~$j$'s message at the end of the first round, and then unconditionally performs~$\act$ at time~1.
We denote by~$r$ the run in which~$\bit=0$. Moreover, let~$r'$ denote the run in which~$\bit=1$ and~$j$ sends the message~$m_j$, while~$r''$ is the run with $\bit=1$ in which the message~$m_j'$ is sent. 

Let~$\varphi$ denote the fact~``$\bit=1$.''
The action $\alpha$ is a deterministic action in~$T$ since agent~$i$ performs it unconditionally at time 1. Hence, from Lemma~\ref{lem:independence} it follows that $\varphi$ is local-state independent of~$\alpha$ in~$T$.
Recall that the value of~$\bit$ does not change over time, and so~$\varphi$ is a fact about runs. Since, by definition, $\muT(\varphi)=\muT(\mbox{``}\bit=1\mbox{''})=p$ and since~$i$ performs~$\act$ in all runs of~$\hat{T}$, we clearly have that 
$\mu_{\hat{T}}(\varphi@\act\Bigiven\,\racti)\,= p$.
Agent~$i$ receives the same message in~$r$ and in~$r'$, and so it has the same local state when it performs~$\act$
in both~$r$ and~$r'$. 
By definition, it follows that $$\big(\pBel_i(\varphi)@\act\big)[r]~~=~~ \big(\pBel_i(\varphi)@\act\big)[r']~~=~~ \frac{\mu_{\hat{T}}(r')}{\mu_{\hat{T}}(\{r,r'\})}~~
~~=~~
\frac{p-\vareps}{1-p+p-\vareps}
~~=~~
\frac{p-\vareps}{1-\vareps}~~~.$$
By assumption $0<\vareps<p<1$, which implies that~ $\frac{p-\vareps}{1-\vareps}<p$. 
 Since~$i$ receives the message~$m_j'$ only in the run~$r''$, in which \mbox{$\bit=1$}, we have that  
 $\big(\pBel_i(\varphi)@\act\big)[r'']=1$, and so~$r''$  is the only run in~$\hat{T}$ for which~$\pBel_i(\varphi)@\act \ge p$.
We thus obtain that 
$\mu_{\hat{T}}(\varphi@\act \Bigiven\,\racti)=\mu_{\hat{T}}(\{r',r''\})=p$,
$~~$ while $~~\mu_{\hat{T}}(\pBel_i(\varphi)@\act\ge p \Bigiven\,\racti)=\muT(r'')=\vareps~.$
The claim follows.
\end{proof}

\section{Relating Probabilistic Constraints and Expected Beliefs}
\label{sec:sec5label}
While Example~\ref{ex:1} showed that it is possible to meet a probabilistic constraint while sometimes acting when the agent's belief does meet the constraint's threshold, \Cref{thm:no-bound} shows that the threshold can be met on an arbitrarily small measure of runs. But the proof of \Cref{thm:no-bound} suggests that, intuitively, the degree of belief needs to meet the threshold  ``on average.'' In this section, we prove our main theorem, which formally captures this intuition. Essentially, we define the {\em expected value} of the agent's degree of belief in~$\varphi$  when it performs~$\act$, and prove that this expected degree must meet the threshold for a probabilistic constraint to be satisfied. 
We define the appropriate notion of expectation as follows
(cf.~\cite{HalpernUncertainty}):

\begin{definition}
\label{BeliefGivenAFactExpectation}
Let $\act$ be a proper action for~$i$ in a pps~$T$. 
The expected degree of $i$'s belief regarding~$\varphi$ when it performs~$\act$, denoted by $\expectation_{\muT}\big(\pBel_i(\varphi)@\act \Bigiven\, \racti\big)$, is: 
$$\expectation_{\muT}\big(\pBel_i(\varphi)@\act \Bigiven\, \racti\big) \quad \eqdef \quad
\mathlarger{\mathlarger{\sum}}_{r\in \RT}{\Big[\muT(r\Bigiven\,\,{\racti})\cdot \big(\pBel_i(\varphi)@\act\big)[r]\Big]} \quad.$$
\end{definition}

The expected degree of belief is precisely the expected value of the random variable $\pBel_i(\varphi)@\act$, conditioned on the fact that~$\act$ is performed at some point in the run. 

Our goal is to show that $\muT(\varphi@\act\Bigiven\,\racti)~=~\expectation_{\muT}
   \big(\pBel_i(\varphi)@\act\Bigiven\,\racti\big)$ holds for every proper action~$\act$. 
   As in the case of \Cref{thm:belief implies prob},
   the claim is not true in general. Again, the issue has to do with mixed actions. For a case in which the claim fails, consider again the system~$T$ depicted in \Cref{fig:dependence} and described in~\Cref{sec:sec4label}.
In this case, however, take the fact of interest to be $\varphi=\acti$. 
It is easy to check that $\muT(\varphi@\act\Bigiven\,\racti)=1$, since~$\varphi$ holds by definition whenever~$\act$ is performed. 
As far as~$i$'s beliefs are concerned, $\expectation_{\muT}\big(\pBel_i(\varphi)@\act\Bigiven\,\racti\bigr)=\frac{1}{2}$, since~$i$'s local state at the initial global state~$g_0$ guarantees that~$\act$ will be performed with probability~$\frac{1}{2}$.
Hence, $\muT(\varphi@\act\Bigiven\,\racti)~\boldsymbol{\ne}~\expectation_{\muT}
   \big(\pBel_i(\varphi)@\act\Bigiven\,\racti\big)$ in this example.
   The source of the problem, as before, is the dependence between~$\varphi$ and~$\acti$. Fortunately, local-state independence is, again, all that we need in order to restore order. We can now show:
 

\begin{restatable}{theorem}{recallMixedExpectationBelief}\label{thm:mixedExpectationOfBeliefTrm}\label{thm:Exp}
Let $\act$ be a proper action for agent~$i$ in a pps~$T$. If~$\varphi$  is local-state independent of~$\act$ in~$T$, then
\begin{equation}\label{eq:mixedExpectationOfBeliefTrm}
   \muT(\varphi@\act\Bigiven\,\racti)~=~\expectation_{\muT}
   \big(\pBel_i(\varphi)@\act\Bigiven\,\racti\big) ~~.
\end{equation}
\end{restatable}


In a precise sense, \Cref{thm:mixedExpectationOfBeliefTrm} provides us with a probabilistic analogue of the knowledge of preconditions principle. The theorem directly implies, in particular, that in order for a system to satisfy a probabilistic constraint of the form $\muT(\varphi@\act\Bigiven\,\racti)\ge p$, the expected probabilistic belief in~$\varphi$ that agent~$i$ should have when it performs~$\act$ must be at least~$p$. 
This is, in fact, a \defemph{necessary and sufficient} condition on beliefs for satisfying a probabilistic constraint. 

\subsection{Jeffrey Conditionalization}
\label{sec:Jeff}
\Cref{thm:mixedExpectationOfBeliefTrm} captures the essence of the connection between probabilistic constraints and probabilistic beliefs in a purely probabilistic system. 
Essentially all of our analysis, except for  \Cref{thm:no-bound}, follows from this result. 
The probabilistic underpinnings of our results are based on well established connections between prior and posterior probabilities, since the natural notion of probabilistic beliefs that corresponds 
to distributed protocols in a pps~$T$ is, as defined in \Cref{sec:Beliefs in Sync sys}, in terms of the posterior probability obtained by conditioning the prior probability $\muT$  
induced by~$T$ on the agent's local state.
Indeed, the proof of  \Cref{thm:mixedExpectationOfBeliefTrm} 
is essentially based on a variant of Jeffrey conditionalization~\cite{Jeffrey}, making use of the properness of the action, and local-state independence of the condition. 
While a detailed proof is given in the Appendix, we now briefly discuss some of the elements underlying the proof.%
\footnote{We thank Joe Halpern for suggesting this view of our analysis.} 

A basic theorem commonly referred to as {\it Jeffrey conditionalization} or the {\it law of total probability}, states roughly that if events $X_1,\ldots,X_n$ form a partition of a state space~$S$, and~$E$ is an event over~$S$, 
 then \[Pr(E)~~=~~Pr(X_1)\!\cdot\! Pr(E\gvn X_1)~+~\cdots~+~Pr(X_n)\!\cdot\! Pr(E\gvn X_n)~.\]
One standard interpretation of this is that $Pr(E)$ is the prior
probability of~$E$ when an experiment is performed, and $X_1,\ldots,X_n$ are its
possible outcomes.  Then $Pr(E|X_i)$ is the posterior probability of $E$
conditional on observing outcome $X_i$.  This is precisely what is used in Monderer and Samet's result  that we quoted  from~\cite{SametMonderer1989} (which is not the main contribution of their paper).

A slight generalization of the above property states  that if~$Y$ is another arbitrary event, then
\[Pr(E | Y) ~~=~~ Pr(X_1|Y)\cdot Pr(E|X_1 \cap Y) ~+~ \cdots~+~ Pr(X_n|Y)\cdot Pr(E| X_n \cap Y)~.\]
This is the mathematical property underlying \Cref{thm:Exp}. The events~$X_i$ correspond to the 
sets of runs in which the action~$\act$ is performed at a particular local state $\ell_i$. 
The event~$Y$ here corresponds to~$\ract$ --- the set of runs in which $\act$ is performed. 

\section{Probable Approximate Knowledge} \label{sec:sec6label}
\Cref{thm:mixedExpectationOfBeliefTrm} formally captures the essential connection between beliefs and actions in probabilistic systems  at which probabilistic constraints are satisfied. 
In particular, it induces a tradeoff between the degree of belief an agent holds regarding~$\varphi$ when it acts, and the probability that it holds such strong belief.
As a corollary of \Cref{thm:mixedExpectationOfBeliefTrm} we can show 

\begin{restatable}{theorem}{recallPSK}\label{thm:mixedExpectationOfBelief}\label{thm:PSK}
Let $\act$ be a proper action for agent~$i$ in a pps~$T$, and let~$\varphi$ be local-state independent of~$\act$ in~$T$. For all $\delta,\vareps\in (0,1)$,\hspace{3mm} 
if\hspace{2mm} 
$\muT(\varphi@\act\Bigiven\,\racti)~\ge~ 1-\delta\vareps$
\hspace{3mm}then\hspace{3mm}
$\muT\Big(\pBel_i(\varphi)@\act\,\ge\, 1-\vareps~ \Bigiven\, \racti \Big)~\ge~ 1-\delta$.
\end{restatable}


Informally, \Cref{thm:mixedExpectationOfBelief} can be read as stating that if a probabilistic constraint $\muT(\varphi@\act\Bigiven\,\racti)\ge p$  with threshold $p=1-\delta\vareps$ holds, then when the agent acts, she will probably 
(i.e., w.p.\ at least $1-\delta$) have a strong (i.e., $1-\vareps$) degree of belief that~$\varphi$ holds. 
An especially pleasing form of this result is obtained when we set $\delta=\vareps$:

\begin{restatable}{corollary}{recallPSB}\label{thm:PSB}
Let $\act$ be a proper action for agent~$i$ in a pps~$T$, and let~$\varphi$ be local-state independent of~$\act$ in~$T$.
For all $\vareps\ge 0$,
if \hspace{1.5mm} 
$\muT(\varphi@\act\Bigiven\,\racti)\ge1-\vareps^2$
\hspace{1.5mm}then \hspace{1.5mm}
$\muT\Big(\pBel_i(\varphi)@\act\,\ge\, 1-\vareps~ \Bigiven\, \racti \Big)~\ge~ 1-\vareps.$
\end{restatable}

Recall that we originally asked to what (probabilistic) extent satisfying a probabilistic constraint with threshold~$p$ required the agent to have a degree of belief that meets the threshold when she acts, 
and discovered in \Cref{thm:no-bound} that the threshold can be met with arbitrarily small probability. This corollary provides a positive result, with a slightly relaxed threshold. 
It implies that in order to satisfy a constraint with threshold~$p$, the condition~$\varphi$ must be believed with degree at least $p'$ with probability at least $p'$ for a value of $\,p'=1-\sqrt{1-p}$. 

We view Corollary~\ref{thm:PSB} as showing  that in a system~$T$ that satisfies a probabilistic constraint with a threshold~$p$ that is sufficiently close to~1,  
the constraint's condition~$\varphi$ must be {\em probably approximately known}. 
Recall that the protocol~$\FS$ in Example~\ref{ex:1} satisfies that 
$\mu(\both\gvn\fireA)\ge 0.99$. 
Corollary~\ref{thm:PSB} implies that \defemph{in every protocol that satisfies this  constraint}, 
the probability that Alice's degree of belief that both are firing together when it decides to fire meets or exceeds $0.9$ is at least~$0.9$. 
In many distributed problems, the probabilistic constraints impose   a much higher threshold than~$0.99$. 
If, e.g., the threshold for~$\varphi$  is exponentially close to~1, then, with extremely high probability, the agent must have 
a very strong degree of belief in~$\varphi$  when it acts  (both  exponentially close to~1).   

\section{Discussion} \label{sec:discussion}
We have characterized the properties that an agent's probabilistic beliefs must satisfy when it acts, in order for its behavior to satisfy a probabilistic constraint that 
requires a given condition to hold whp when the agent performs a given action. Our results are not limited to protocols that make explicit reference to the agent's beliefs. 
They apply to all protocols, deterministic and probabilistic, and to arbitrary probabilistic constraints (subject to actions being proper and conditions being local-state 
independent of the actions). In a precise sense, \Cref{thm:Exp} is the probabilistic analogue of the knowledge of preconditions principle, which characterizes a fundamental 
connection between knowledge and action in distributed systems~\cite{KoP-paper}. Just as the KoP has proven useful in the design an analysis of optimal distributed protocols, 
we expect that \Cref{thm:Exp} and its future extensions to provide insights and assist in the design of efficient probabilistic protocols. For a simple example of such an insight, 
observe that \Cref{thm:Exp} implies that whenever an agent acts while having a low degree of belief in the desired condition of a probabilistic constraint, 
she reduces the probability of success. By refraining from doing so, she can improve her performance. 
Thus, for example, even though Alice's actions guarantee success with probability $0.99$ in the $\FS$ protocol, she can satisfy an even more stringent requirement by 
avoiding to fire when she receives a $\NO$ message from Bob. The probability that both fire, given that Alice fires, goes up to $0.99899$. 
Moreover, if an agent never acts when her degree of belief is below the threshold, 
\Cref{thm:Exp} can be used to establish that an agent's actions are optimal with respect to satisfying a probabilistic constraint, given her information.

%
\bibliographystyle{plain}
\bibliography{PKoPBib}

\appendix

\section{
Notations and Observations 
for the Proofs
}\label{AppendixProofs}

Sections of the Appendix are devoted to providing detailed proofs of all technical claims in the paper. 
We start by defining notation and stating several observations that will be used in various sections of the Appendix. 

\vspace{2mm}
Given an action $\act\in\Acti$ in a pps~$T$, we use $\Ract$ to denote event corresponding to the fact~$\ract$, i.e., the set of runs in which~$\act$ is performed. More formally, 
\[\Ract\eqdef\{r\in\RT|\,r\sat\racti\}~~.\]
%
Several proofs will partition the event~$\Ract$, according to the local states at which~$i$ performs the actions. To this end, we will use $\Li[\act]$ to denote the set of local states at which~$i$ ever performs~$\act$. That is, 
\[\Li[\act]~\eqdef~\{\ell_i=r_i(t)\,|\, (T,r,t)\models\acti\}~~.
\]
%
For ease of exposition we will use ${\act@\li}$ as shorthand for ${\acti@\li}$. 
By their definition, facts about runs such as $\varphi@\li$ and $\act@\li$ hold only in runs in which agent~$i$'s local state is~$\li$ at some point. 
Similarly, $(T,r)\models \act@\li$  only if, in particular, $ (T,r) \models\racti$ (that is, only if~$i$ performs~$\act$ in~$r$).
We will make use of the following trivial consequences of these observations to simplify expressions in the proofs.\\[-2ex]
\begin{lemma}\label{clm: claim7}
    For all~$T$, $r$, $\li$ and~$\act$, the following equivalences hold.  
   \begin{itemize}
   \item[(a)] $\quad(T,r)\models~~\big(\act@\li ~\Leftrightarrow~\act@\li\wedge\rli\big)$~,
\item[(b)] $\quad(T,r)\models~~\big( ~[\varphi\wedge\act]@\li ~\Leftrightarrow~~[\varphi\wedge\act]@\li\wedge\rli\big)$~,
\item[(c)] $\quad(T,r)\models~~\big( ~[\varphi\wedge \act]@\li ~\wedge~ \act@\li  ~\Leftrightarrow~~[\varphi\wedge \act]@\li \big)$~,
\item[(d)] $\quad (T,r)\models~~\big(\act@\li ~\Leftrightarrow~\act@\li\wedge\racti\big)$~, ~~and
\item[(e)] $\quad (T,r)\models~~\big(\varphi@\act ~\Leftrightarrow~\varphi@\act \wedge\racti\big)$~.
   \end{itemize}
\end{lemma}

The appendix section is organized as follows.
We prove \Cref{thm:belief implies prob} and Lemma~\ref{lem:independence} in \Cref{Appendix:SecB} and \Cref{Appendix:SecC}, respectively.
Our main result, stated in \Cref{thm:mixedExpectationOfBeliefTrm}, is proved in \Cref{Appendix:SecD}, and is used to prove the later theorems and lemmas.
Lemma~\ref{lem:sometimes} is proved in \Cref{Appendix:SecE}.
Finally, \Cref{thm:PSK} and Corollary~\ref{thm:PSB} are proved in \Cref{Appendix:SecF}.

\setcounter{equation}{0}
\section{Proving Theorem~4.2}
\label{Appendix:SecB}

Roughly speaking,  \Cref{thm:belief implies prob}
shows that if a protocol ensures that actions be taken only when belief meets the threshold of a probabilistic constraint, then the probabilistic constraint will be satisfied. This is true, however, only if the constraint's condition and its action satisfy local-state independence. Before proving the theorem, we prove a lemma showing that local-state independence guarantees that
 roughly speaking, the probability that~$\varphi$ holds at a point at which $i$'s local state is in~$\Li[\act]$ is independent of whether~$\act$ is performed at that state. This is where the independence property is used in the proof of \Cref{thm:belief implies prob}. More formally:
\begin{lemma}\label{lem:indep statement}
Let $\act$ be a proper action for~$i$ in~$T$, and let $\varphi$ be local-state independent of~$\newacti$ in~$T$. Then for each $\li\in\Li[\act]$,
\begin{equation}\label{eq:phi at alpha equals belief}
    \muT(\varphi@\act\Bigiven\, \act@\li) \quad=\quad
\muT(\varphi@\li\Bigiven\, \rli)~.
\end{equation}
\end{lemma}
\begin{proof}
We first define a partition on the set of runs in which~$\varphi$ occurs when~$\act$ is performed. 
For every local state $\ell\in\Li[\act]$, we denote  $Q_{\varphi}^{\ell}=\{r|\; (T,r) \models ~[\varphi \wedge \act]@\ell ~\}$. 
Thus, $Q_{\varphi}^{\ell}$ is the set of runs in which both~$\varphi$ holds and~$\act$ is performed when~$i$'s local state is~$\ell$. 
Moreover, define $\Pi'=\{ Q_{\varphi}^{\ell}\,|\, \ell\in\Li[\act] \}$. 

Fix a local state $\li\in\Li[\act]$.
Clearly,
$\Pi'$ partitions the set of runs satisfying $\varphi@\act$, and so the left-hand side of~\eqref{eq:phi at alpha equals belief} satisfies: 
\begin{equation}\label{eq:lsi1}
    \muT(\varphi@\act\Bigiven\, \act@\li) \quad=\quad
    \mathlarger{\mathlarger{\sum}}_{\lip\in \Li[\act] }{ \muT \big([\varphi\wedge \act]@\lip\Bigiven\, \act@\li \big)} ~.
\end{equation}
By assumption, $\act$ is a proper action, hence there exists no run~$r\in\RT$ for which both $(T,r)\models \act@\li$ and $(T,r)\models \act@\lip$ for $\lip\neq\li$. It follows that $\muT \big([\varphi\wedge \act]@\lip\Bigiven\, \act@\li \big)=0$ for each $\lip\neq\li$, and so the right-hand side of~\eqref{eq:lsi1} satisfies:
\begin{equation}\label{eq:lsi2}
    \mathlarger{\mathlarger{\sum}}_{\lip\in \Li[\act] }{ \muT \big([\varphi\wedge \act]@\lip\Bigiven\, \act@\li \big)} 
    \quad=\quad
    \muT \big([\varphi\wedge \act]@\li\Bigiven\, \act@\li \big) ~.
\end{equation}
From the definition of conditional probability we obtain that the right-hand side of~\eqref{eq:lsi2} satisfies:
\begin{equation}\label{eq:lsi3}
    \muT \big([\varphi\wedge \act]@\li\Bigiven\, \act@\li \big) 
    \quad=\quad
    \frac{\muT \big([\varphi\wedge \act]@\li ~\wedge~ \act@\li \big) }{\muT (\act@\li )} ~.
\end{equation}
From Lemma~\ref{clm: claim7}(c)
it follows that $\muT \big([\varphi\wedge \act]@\li ~\wedge~ \act@\li \big) ~=~ \muT \big([\varphi\wedge \act]@\li \big)$.
By multiplying the right-hand side of~\eqref{eq:lsi3} by~$\frac{\muT(\rli)}{\muT(\rli)}=1$, we obtain: 
\begin{equation}\label{eq:lsi4}
    \frac{\muT \big([\varphi\wedge \act]@\li ~\wedge~ \act@\li \big) }{\muT (\act@\li )}
    \quad=\quad
    \frac{\muT \big([\varphi\wedge \act]@\li \big) }{\muT (\rli)} \cdot \frac{\muT (\rli)}{\muT (\act@\li )} ~.
\end{equation}
From Lemma~\ref{clm: claim7}(b)
Lemma~\ref{clm: claim7}(a)
and by the definition of conditional probability we obtain  that the right-hand side of~\eqref{eq:lsi4} satisfies:
\begin{equation}\label{eq:lsi5}
    \frac{\muT \big([\varphi\wedge \act]@\li \big) }{\muT (\rli)} \cdot \frac{\muT (\rli)}{\muT (\act@\li )}
    \quad=\quad
    \frac{\muT \big([\varphi\wedge \act]@\li \Bigiven\,\rli \big) }{\muT (\act@\li \Bigiven\,\rli)} ~.
\end{equation}
By assumption, $\varphi$ is local-state independent of~$\act$ in~$T$, and so the right-hand side of~\eqref{eq:lsi5} satisfies:
\begin{equation}\label{eq:lsi6}
    \begin{aligned}
    \frac{\muT \big([\varphi\wedge \act]@\li \Bigiven\,\rli \big) }{\muT (\act@\li \Bigiven\,\rli)}
    \quad=\quad
    \frac{\muT (\varphi@\li \Bigiven\,\rli ) \cdot \muT (\act@\li \Bigiven\,\rli ) }{\muT (\act@\li \Bigiven\,\rli)}
    \quad=\quad 
   \muT (\varphi@\li \Bigiven\,\rli ) ~.
    \end{aligned}
\end{equation}

\end{proof}

We are now ready to prove: 
\recallBeliefImpliesProb*

\begin{proof}
Recall that $\Ract\eqdef\{r\in\RT|\,r\sat\racti\}$ and $\Li[\act]~\eqdef~\{\ell_i=r_i(t)\,|\, (T,r,t)\models\acti\}$.
Moreover, we use ${\act@\li}$ as shorthand for ${\acti@\li}$.
For every $\li\in\Li[\act]$, let $\Qli=\{r|\,(T,r) \!\models\act@\li\}$, and let $\Pi=\{ \Qli| \li\in\Li[\act] \}$.
We claim that $\Pi$ is a partition of~$\Ract$. 
Every set $\Qli\in\Pi$ is nonempty since by the definition $\li\in\Li[\act]$ and so there is at least one run of~$\RT$ in which~$\li$ appears as a local state of~$i$. 
By definition, $\bigcup\limits_{\li\in\Li[\act]} \Qli=\Ract$ since for every $r\in\RT$ we have that  $r\in\Ract$ 
iff $i$ performs~$\act$ at some point in~$r$, which is true iff $i$ performs~$\act$ at a local state of~$\Li[\act]$ at some point $(r,t)$ of~$r$. 
Finally, the intersection of any two sets $\Qli$ and $Q^{\lip}$ for $\li,\lip \in \Li[\act]$ such that $\li\neq\lip$ is empty since, by assumption,  $\act$ is performed at most once in any run of $\RT$.
Thus, $\Pi$ is a partition of~$\Ract$. 


\vspace{2.5mm}
Since~$\Pi=\{ \Qli| \li\in\Li[\act] \}$ is a partition of~$\Ract$, we can use the law of total probability to obtain:
\begin{equation}\label{eq:law of total probability}
    \muT(\varphi@\act\Bigiven\,\racti)\quad=\quad 
\mathlarger{\mathlarger{\sum}}_{\li\in \Li[\act] }{ \muT(\varphi@\act\Bigiven\, \act@\li \wedge \racti) \cdot \muT(\act@\li \Bigiven\, \racti) }.
\end{equation}
From Lemma~\ref{clm: claim7}(d)
it follows that $$\muT(\varphi@\act\Bigiven\, \act@\li \wedge \racti)~~=~~\muT(\varphi@\act\Bigiven\, \act@\li)$$ for every $\li\in \Li[\act]$. It follows that the right-hand side of \eqref{eq:law of total probability} equals:
\begin{equation}\label{eq:rli iff (rli wedge racti)}
\mathlarger{\mathlarger{\sum}}_{\li\in \Li[\act] }{ \muT(\varphi@\act\Bigiven\, \act@\li) \cdot \muT(\act@\li \Bigiven\, \racti) } \quad.
\end{equation}
Fix $\li\in\Li[\act]$. By definition of $\Li[\act]$ there exists a point $(r,t)\in \Pts(T)$ such that $(T,r,t)\models \acti$, and $\li= \rit$. 
By assumption we have that $(T,r,t)\models \pBel_i(\varphi)\ge p$, 
which implies by definition of $\pBel_i$ that $\muT(\varphi@\li\given\rli) \ge p$. 
Moreover, by Lemma~\ref{lem:indep statement}, it follows that $\muT(\varphi@\act\Bigiven\, \act@\li) \ge p$. 
Since this is true for every $\li\in\Li[\act]$, we thus obtain that:
$$  \eqref{eq:rli iff (rli wedge racti)}
\quad\ge \quad
p~\cdot \mathlarger{\mathlarger{\sum}}_{\li\in \Li[\act]}{ \muT(\act@\li \given \racti) }\quad=\quad p\cdot 1~~=~~p~.$$
It follows that $\muT(\varphi@\act\Bigiven\,\racti) \ge p$ as required.
\end{proof}

\section{Proving Lemma~4.3}
\label{Appendix:SecC}
We recall Lemma~\ref{lem:independence}: 
\vspace{4mm}
\newline
\textbf{Lemma 4.3}\quad
{\it Let~$\act\in\Acti$ be a proper action in a pps~$T$, and let~$\varphi$ be a fact over~$T$. If (a) $\act$ is deterministic in~$T$, or 
(b) $\varphi$ is past-based in~$T$, then 
$\varphi$ is local-state  independent of~$\act$ in~$T$. }


\begin{proof}
We Prove each condition separately:
\paragraph{Case (a): $\act$ is a deterministic action in~$T$.}
Fix~$\varphi$ and an action~$\act$ satsfying the assumptions, and a local state $\li\in \Li[\act]$. Since~$\act$ is a deterministic proper action for~$i$, the fact $\act@\li$ either holds for every run that satisfies~$\rli$ (i.e., for every run in which~$\li$ appears), 
or holds for none of them.
\begin{itemize}[leftmargin=.1in]
    \item[-] If~$\act@\li$ holds for every run that satisfies~$\rli$ then, for every $r\in\RT$ we have both that\newline 
$(T,r) \models ~\act@\li \Leftrightarrow \rli$, and that
$(T,r) \models ~\big([\varphi@\li \Leftrightarrow [\varphi\wedge\act]@\li\big)$.
We thus obtain that $\muT(\act@\li\Bigiven\rli)=1$, and that
$
\muT\big(\varphi@\li~ \Bigiven ~\rli\big) = 
\muT\big([\varphi\wedge \act]@\li~ \Bigiven ~\rli \big)$.
\item[-]
Otherwise, $\act@\li$ holds at none of the runs that satisfies~$\rli$. In this  case we have for all runs $r\in\RT$ that both 
$(T,r) \models \neg(\act@\li) ~\Leftrightarrow~ \rli$, and
$(T,r) \models \big(~ \neg\big(~[\varphi\wedge\act]@\li \big) ~\Leftrightarrow~ \varphi@\li ~\big)$.
We thus have that both
$\muT(\act@\li\Bigiven\rli)=0$, and
$\muT\big( ~[\varphi\wedge \act]@\li~ \Bigiven ~\rli \big) ~=~ 0$.
\end{itemize}
In either case
$\muT(\varphi@\li\Bigiven\,\rli) \cdot \muT(\act@\li\Bigiven\rli)~=~\muT\big([\varphi\wedge \act]@\li \Bigiven\rli \big)$, establishing local-state independence.

\paragraph{Case (b): $\varphi$ is a past-based fact in~$T$.}
Fix $\li\in \Li[\act]$, and let $p=\muT(\varphi@\li\Bigiven\,\rli)$, thus $p$ denotes the probability that $\varphi$ holds when $i$'s local state is $\li$ given that $i$'s local state is $\li$ at some point of the run.
Since~$\varphi$ is a past-based fact in~$T$, for each node~$v$ of~$T$ either~$\varphi$ holds at any point~$(r,t)$ such that $r$~passes through~$v$ at time~$t$ or~$\varphi$ does not hold at any such point.

Since~$i$'s protocol~$\Proti$ is a (possibly probabilistic) function of its local state, 
the probability that $i$~performs~$\act$ is the same at all points at which its local state is~$\li$. Hence, the conditional probability that $i$~performs~$\act$ when its local state is~$\li$ given that~$i$'s local state is~$\li$ at some point of the run, denoted $\muT(\act@\li\Bigiven\rli)$, is fixed by the protocol. 
%
It follows that 
$\muT\big( ~[\varphi \wedge \act]@\li \Bigiven ~\rli \big)$ is the probability of reaching a node $v$ at which~$\varphi$ holds and then choosing a run in which~$i$ performs~$\act$ at $v$.
By the analysis above it equals $\muT(\varphi@\li\Bigiven\,\rli) \cdot \muT(\act@\li\Bigiven\rli)$, and
the claim follows.
\end{proof}

%

%

\section{Proving 
 the Expectation  Theorem }\label{Appendix:SecD}
The expectation theorem 
is our main result, and the proofs of other claims, including Lemma~\ref{lem:sometimes}, easily follow as its corollaries. 
It is stated as follows: 
\recallMixedExpectationBelief*

\begin{proof}
We will transform the right hand side of Equation~\eqref{eq:mixedExpectationOfBeliefTrm} into the left hand side. 
%
Recall that $\Ract$ is the set of runs of~$T$ in which~$i$ performs the action~$\act$. 
Since $\muT(r\given\,{\racti})=0$ for each run $r\notin \Ract$,  the expected degree of~$i$'s belief  in~$\varphi$ when it performs a proper action~$\act$ (given that~$i$ performs~$\act$) can be expressed by: 
\begin{equation}
\label{eq:Ract}
\expectation_{\muT}(\pBel_i(\varphi)@\act\Bigiven\,\racti)~~~=~~~
\mathlarger{\mathlarger{\sum}}_{r\in \Ract}{\Big[\muT(r\Bigiven\,\,{\racti})\cdot \big(\pBel_i(\varphi)@\act\big)[r]\Big]}~~~.
\end{equation}
Recall that  $\Li[\act]~=~\{\ell_i= r_i(t)\,|\, (T,r,t)\models\acti\}$ for every $\act\in\Acti$, 
and  $\Qli=\{r|\, (T,r) \!\models\act@\li\}$ for every $\li\in\Li[\act]$. Moreover, $\Pi=\{ \Qli\,|~ \li\in\Li[\act] \}$ is a partition of~$\Ract$.
The rhs of Equation~\eqref{eq:Ract} can thus be reformulated as: 
\begin{equation}\label{eq:Eq1.1}
 \mathlarger{\mathlarger{\sum}}_{\li\in \Li[\act]}{ ~~\mathlarger{\mathlarger{\sum}}_{\substack{r\in \Qli} }{ \Big[
\muT(r\given\racti)\cdot \big(\pBel_i(\varphi)@\act\big)[r] \Big] } } ~~.
\end{equation}
Since~$\act$ is 
a proper action, and so if $(T,r) \models \act@\li$ for a local state~$\li\in\Li[\act]$, then $\big(\pBel_i(\varphi)@\act\big)[r]=\muT(\varphi@\li \Bigiven\,\rli)$.
Moreover, $\muT(\varphi@\li \Bigiven\,\rli)$ is guaranteed to be well defined since by definition of $\Li[\act]$, for each $\li\in\Li[\act]$ there must be a run $r$ such that $(T,r) \models\rli$, where $\muT$ assigns positive measure to every run of~$\RT$, and so $\muT(\rli)>0$ for each $\li\in\Li[\act]$.
We can therefore rewrite Equation~\eqref{eq:Eq1.1} to obtain:
\begin{equation}\label{eq:ExpSplit}
   \mathlarger{\mathlarger{\sum}}_{\li\in \Li[\act]}{ ~~\mathlarger{\mathlarger{\sum}}_{\substack{r\in \Qli} }{ \Big[
    \muT(r\given\racti)\cdot \muT(\varphi@\li \Bigiven\,\rli) \Big] } } ~~. 
\end{equation}
Since $\muT(\varphi@\li \Bigiven\,\rli)$ is a constant given that~$r\in \Qli$, an equivalent form of \eqref{eq:ExpSplit} is: 
\begin{equation}\label{eq:Eq4}
 \mathlarger{\mathlarger{\sum}}_{\li\in \Li[\act]}{ \Bigg[   \muT(\varphi@\li\Bigiven\,\rli) \cdot \mathlarger{\mathlarger{\sum}}_{\substack{r\in \Qli} }{ \muT(r\given\racti) } \Bigg] }~~. 
 \end{equation}
The inner summation is performed over the conditional probabilities of the distinct runs whose union is $\Qli$, and thus equals $\muT(\act@\li \given\racti)$. 
Equation~\eqref{eq:Eq4} can thus be rewritten as: 
\begin{equation}\label{eq:Eq5}
\mathlarger{\mathlarger{\sum}}_{\li\in \Li[\act]}{ \Bigg[   \muT(\varphi@\li\Bigiven\,\rli) \cdot \muT(\act@\li \given\racti) \Bigg] }~~.
\end{equation}
Applying the definition of conditional probability  to the right element of each summand, Equation~\eqref{eq:Eq5} becomes: 
\begin{equation}\label{eq:Eq6}
\frac{1}{\muT(\racti)}~\cdot \mathlarger{\mathlarger{\sum}}_{\li\in \Li[\act]}{ \Bigg[   \muT(\varphi@\li\Bigiven\,\rli) \cdot \muT(\act@\li \wedge\racti) \Bigg] }~~.
\end{equation}
%
By Lemma~\ref{clm: claim7}(d)
we can rewrite the right element of each summand in \eqref{eq:Eq6} to obtain:
\begin{equation}\label{eq:Eq7}
\frac{1}{\muT(\racti)}~\cdot \mathlarger{\mathlarger{\sum}}_{\li\in \Li[\act]}{ \Bigg[   \muT(\varphi@\li\Bigiven\,\rli) \cdot \muT(\act@\li) \Bigg] }~~.
\end{equation}
Recall that $\muT(\rli)>0$ for $\li\in\Li[\act]$.
We multiply each summand in Equation~\eqref{eq:Eq7} by $\frac{\muT(\rli)}{\muT(\rli)}=1$ and obtain:
\begin{equation}\label{eq:Eq8}
\frac{1}{\muT(\racti)}~\cdot \mathlarger{\mathlarger{\sum}}_{\li\in \Li[\act]}{ \Bigg[   \muT(\varphi@\li\Bigiven\,\rli) \cdot \frac{\muT(\act@\li)}{\muT(\rli)} \cdot \muT(\rli) \Bigg] }~~.
\end{equation}
By Lemma~\ref{clm: claim7}(a)
and from the definition of conditional probability we rewrite the second element of each summand to obtain:
\begin{equation}\label{eq:Eq9}
\frac{1}{\muT(\racti)}~\cdot \mathlarger{\mathlarger{\sum}}_{\li\in \Li[\act]}{ \Bigg[   \muT(\varphi@\li\Bigiven\,\rli) \cdot \muT(\act@\li\Bigiven\rli) \cdot \muT(\rli) \Bigg] }~~.
\end{equation}
Since~$\varphi$ is local-state independent of~$\act$ in $T$, we have  by \Cref{def:local-state indep} that $\muT(\varphi@\li\Bigiven\,\rli) \cdot \muT(\act@\li\Bigiven\rli)=\muT\big([\varphi\wedge \act]@\li \Bigiven\rli \big)$.
Hence, we can simplify Equation~\eqref{eq:Eq9} into: 
\begin{equation}\label{eq:Eq10}
\frac{1}{\muT(\racti)}~\cdot \mathlarger{\mathlarger{\sum}}_{\li\in \Li[\act]}{ \Bigg[   \muT\big( ~[\varphi \wedge \act]@\li ~\Bigiven\rli \big) \cdot \muT(\rli) \Bigg] } ~~.
\end{equation}
%
By Lemma~\ref{clm: claim7}(b) and 
 the definition of conditional probability we obtain:
\begin{equation}\label{eq:Eq11}
\frac{1}{\muT(\racti)}~\cdot \mathlarger{\mathlarger{\sum}}_{\li\in \Li[\act]}{ \Bigg[   \muT\big( ~[\varphi \wedge \act]@\li ~\big) \Bigg] }~~.
\end{equation}
Recall that $Q_{\varphi}^{\li}$ is the set of runs in which both~$\varphi$ holds and~$\act$ is performed when~$i$'s local state is~$\li$. 
By definition, we have that $\muT(Q_{\varphi}^{\li}) = \muT\big( ~[\varphi \wedge \act]@\li ~\big)$, and so we can transform Equation~\eqref{eq:Eq11} into: 
\begin{equation}\label{eq:Eq12}
\frac{1}{\muT(\racti)}~\cdot \mathlarger{\mathlarger{\sum}}_{\li\in \Li[\act]}{ \Big[   \muT(Q_{\varphi}^{\li}) \Big] }~~.
\end{equation}
%
Define $\Pi'\triangleq\{ Q_{\varphi}^{\li}\,|\, \li\in\Li[\act] \}$. Clearly, $\Pi'$ is a partition of the runs satisfying $\varphi@\act$. 
We can therefore rewrite Equation~\eqref{eq:Eq12} as: %
\begin{equation}\label{eq:Eq13}
\frac{1}{\muT(\racti)}~\cdot \mathlarger{\mathlarger{\sum}}_{Q_{\varphi}^{\li}\in \Pi'}{ \Big[   \muT(Q_{\varphi}^{\li}) \Big] }~~.
\end{equation}
Since~$\Pi'$ is a partition of the runs satisfying $\varphi@\alpha$, we can further rewrite Equation~\eqref{eq:Eq13} into: 
\begin{equation}\label{eq:Eq14}
\frac{1}{\muT(\racti)}\cdot~ \muT\big(\varphi@\act \big) ~~.
\end{equation}
From Lemma~\ref{clm: claim7}(e)
and the definition of conditional probability, the expression in \eqref{eq:Eq14} equals 
$\muT\big(\varphi@\act \Bigiven \racti \big)$.
It follows that 
$$\muT(\varphi@\act\Bigiven\,\racti)~=~\expectation_{\muT}(\pBel_i(\varphi)@\act\Bigiven\,\racti)~~,$$ 
as claimed. 
\end{proof}

\section{Proving Lemma~5.1}
\label{Appendix:SecE}
We recall Lemma~\ref{lem:sometimes}: 
\vspace{4mm}
\newline
\textbf{Lemma~5.1}~~ 
{\it Let $\act$ be a proper action for agent~$i$ in a pps~$T$, and let a fact~$\varphi$ be local-state independent of~$\act$ in~$T$.
If \mbox{$\muT(\varphi@\act\Bigiven\,\racti)\ge p$}, then there must be at least one point $(r,t)$ of~$~T$ at which $\act$ is performed and $(T,r,t)\sat \,\,\pBel_i(\varphi)\ge p$.}
%
\begin{proof}
We prove the counterpositive. 
Recall \Cref{eq:Ract}:
$$\expectation_{\muT}(\pBel_i(\varphi)@\act\Bigiven\,\racti)=
\mathlarger{\mathlarger{\sum}}_{r\in \Ract}{\Big[\muT(r\given\racti) \cdot \big(\pBel_i(\varphi)@\act\big)[r] \Big]}$$
Begin with the right-hand side of~\eqref{eq:Ract} and assume that $\pBel_i(\varphi)<p$ whenever~$i$ performs~$\act$, thus:
$$ \mathlarger{\mathlarger{\sum}}_{r\in \Ract}{\Big[\muT(r\given\racti) \cdot \big(\pBel_i(\varphi)@\act\big)[r] \Big]} \quad < \quad
\mathlarger{\mathlarger{\sum}}_{r\in \Ract}{\Big[\muT(r\given\racti) \cdot p \Big]}\quad = \quad 
p~\cdot \mathlarger{\mathlarger{\sum}}_{r\in \Ract}{ \muT(r\given\racti) } ~~=~~ p\quad.$$
Hence, $\expectation_{\muT}(\pBel_i(\varphi)@\act\Bigiven\,\racti) < p$.
Recall from \Cref{thm:mixedExpectationOfBeliefTrm} that $$\muT(\varphi@\act\Bigiven\,\racti)=\expectation_{\muT}(\pBel_i(\varphi)@\act\Bigiven\,\racti) ~~,$$ and so if $\pBel_i(\varphi)<p$ whenever~$i$ performs~$\act$, then~$\muT(\varphi@\act\Bigiven\,\racti)<p$ .
The claim follows.
\end{proof}

\section{Proving Theorem~7.1 and Corollary~7.2} 
\label{Appendix:SecF}

We now turn to proving 
Theorem~7.1 and Corollary~7.2, which show, roughly, that if~$\varphi$ is guaranteed to hold whp when~$i$ performs~$\act$, then the agent must probably approximately know that~$\varphi$ holds when it performs~$\act$.
We start by proving Lemma~\ref{lem:probability1}, which essentially establishes the claim when the threshold probability is~1.
%
In this case, under our model assumptions, if~$\varphi$ must hold when~$i$ performs an action~$\act$, then whenever~$i$ acts it must {\em know} that~$\varphi$ currently holds.
This claim is stated formally in Lemma~\ref{lem:probability1}, which closely corresponds to the claim made by \KoP:
\begin{restatable}{lemma}{recallProbabilityOne}\label{lem:probability1}
    Let $\act$ be a proper action for agent~$i$ in a pps~$T$, and let a fact~$\varphi$ be local-state independent of~$\act$ in~$T$.
    If ~$\muT(\varphi@\act\Bigiven\,\racti)=1$~ then 
    ~$
    \muT\Big(\pBel_i(\varphi)@\act~=~ 1\, \Bigiven\, \racti \Big)~=~ 1.$
\end{restatable}

\begin{proof}
We prove the counterpositive. 
Recall that $\Qli=\{r|\, (T,r) \!\models\act@\li\}$ for every $\li\in\Li[\act]$, and that $\Pi=\{ \Qli\,| ~\li\in\Li[\act] \}$ is a partition of~$\Ract$.
From \eqref{eq:ExpSplit} 
we obtain:
$$ \muT(\varphi@\act \given \racti) ~=~ 
\mathlarger{\mathlarger{\sum}}_{\li\in \Li[\act]}{ ~~\mathlarger{\mathlarger{\sum}}_{\substack{r\in \Qli} }{ \Big[
\muT(r\given\racti)\cdot \muT(\varphi@\li \Bigiven\,\rli) \Big] } } \quad. $$

Assume that there exists a run~$r\in\Ract$ in which~$i$'s belief in~$\varphi$ when it performs~$\act$ is smaller than~1, i.e.,  \mbox{$\big(\pBel_i(\varphi)@\act\big)[r] < 1$}.
Then there exists a point~$(r,t)\in\Pts(T)$ at which~$i$ performs~$\act$, but ${(T,r,t)\models~ \pBel_i(\varphi)<1}$. For the state $\li=\rit$ we obtain by definition of~$\pBel_i(\varphi)$ that $\muT(\varphi@\li \given\rli)<1$. Thus:

\begin{equation}\label{eq: ineq prob1 .1}
\begin{split}
    \mathlarger{\mathlarger{\sum}}_{\li\in \Li[\act]}{ ~~\mathlarger{\mathlarger{\sum}}_{\substack{r\in \Qli} }{ \Big[
\muT(r\given\racti)\cdot \muT(\varphi@\li \Bigiven\,\rli) \Big] } } \quad < \quad 
\mathlarger{\mathlarger{\sum}}_{\li\in \Li[\act]}{ ~~\mathlarger{\mathlarger{\sum}}_{\substack{r\in \Qli} }{ \Big[
\muT(r\given\racti)\cdot 1 \Big] } } \quad.
\end{split}
\end{equation}
Since~$\Pi$ is a partition of~$\Ract$, the right-hand side of~\eqref{eq: ineq prob1 .1} equals:
$$ \mathlarger{\mathlarger{\sum}}_{\substack{r\in \Ract} }{ \Big[
\muT(r\given\racti) \Big] } \quad = \quad 1 \quad.$$
%
It follows that if there exists a run~$r\in\Ract$ such that~$i$'s belief in~$\varphi$ when it performs~$\act$ is smaller than~1, then~ \mbox{$\muT(\varphi@\act \Bigiven\, \racti) < 1$}. This establishes the counterpositive claim, and completes the proof.
\end{proof}

%
 
 \recallPSK*
 
 \begin{proof}
Recall that~$\Ract$ denotes the set of runs in which~$i$ performs~$\act$. 
First note that if $\big(\pBel_i(\varphi)@\act\given\racti\big)[r]\ge 1-\vareps$ for all $r\in\Ract$, then 
$\muT\Big(\pBel_i(\varphi)@\act\,\ge\, 1-\vareps~ \Bigiven\, \racti \Big)~=~1~\ge~ 1-\delta$, and we are done. 
From here on, the proof will be performed under the assumption that the set of runs $r\in\Ract$ for which 
$\big(\pBel_i(\varphi)@\act\given\racti\big)[r]< 1-\vareps$ ~is not empty (and thus has positive probability).

\Cref{eq:Ract} states the following:
\[    \expectation_{\muT}\big(\pBel_i(\varphi)@\act\given\racti\big)~~=~~~
\mathlarger{\mathlarger{\sum}}_{r\in \Ract}{\Big[\muT(r\given\,{\racti})\cdot \big(\pBel_i(\varphi)@\act\big)[r]\Big]} \quad .\]
We can partition the runs of~$\Ract$ into ones in which $\pBel_i(\varphi)<1-\vareps$ when~$i$ performs~$\act$, and ones in which $\pBel_i(\varphi)\ge 1-\vareps$ there. 
Since~$\varphi$ and $\act$ are fixed throughout the proof, we will use the following shorthands for ease of exposition. 
We denote by $R_{<1-\vareps}$ the set of runs of $\Ract$ for which $(\pBel_i(\varphi)@\act) < 1-\vareps$; similarly, we use $R_{\ge 1-\vareps}$ to denote the set of runs of $\Ract$ for which $(\pBel_i(\varphi)@\act) \ge 1-\vareps$.
Thus $\eqref{eq:Ract}$ can be reformulated as:
\begin{equation}\label{eq:bel split}
\begin{aligned}
\expectation_{\muT}\big(\pBel_i(\varphi)@\act\given\racti & \big)  ~~=~~ \\
    &\mathlarger{\mathlarger{\sum}}_{r\in R_{< 1-\vareps} }{   \Big[\muT(r\given\,{\racti})\cdot  \big(\pBel_i(\varphi)@\act\big)[r]\Big]} ~+ 
     ~ \mathlarger{\mathlarger{\sum}}_{r\in R_{\ge 1-\vareps} }{\Big[\muT(r\given\,{\racti})\cdot \big(\pBel_i(\varphi)@\act\big)[r]\Big]}
     \quad.
\end{aligned}
\end{equation}
\newline
We remark that in runs in which $\pBel_i(\varphi)@\act \ge 1-\vareps$, agent $i$'s belief is upper bounded by~1. In addition, recall that, by assumption, 
$R_{<1-\vareps}$ is not an empty set. Hence,  the right-hand side of \eqref{eq:bel split} satisfies:
\begin{equation}\label{eq:A15}
    \begin{aligned}
    \mathlarger{\mathlarger{\sum}}_{r\in R_{< 1-\vareps} }{\Big[\muT(r\given\,{\racti})\cdot \big(\pBel_i(\varphi)@\act\big)[r]\Big]} ~~&+~ 
  \mathlarger{\mathlarger{\sum}}_{r\in R_{\ge 1-\vareps} }{\Big[\muT(r\given\,{\racti})\cdot \big(\pBel_i(\varphi)@\act\big)[r]\Big]}~~< \\
    (1-\vareps)\cdot \muT\big(\pBel_i(\varphi)@\act < 1-\vareps\, \given \racti\big)  \; &+\; ~
  1\cdot \muT\big(\pBel_i(\varphi)@\act \ge 1-\vareps\, \given \racti\big) \quad.
    \end{aligned}
\end{equation}

We will prove the counterpositive.
Assume that 
$\muT\big(\pBel_i(\varphi)@\act \ge 1-\vareps\, \given \racti\big) <1-\delta $. Thus, for some $\delta'>\delta$, 
\mbox{$\muT\big(\pBel_i(\varphi)@\act \ge 1-\vareps\, \given \racti\big) =1-\delta' $} , and the probability 
$\muT\big(\pBel_i(\varphi)@\act < 1-\vareps\, \given \racti\big)$ of the complementary event is $\delta'$. It follows that:
$$\eqref{eq:A15} \quad=\quad (1-\vareps)\cdot \delta' \;+\; 1\cdot(1-\delta') \quad = \quad
1-\vareps \delta'\quad.$$

By assumption, 
$\delta' >\delta $ and so  $~\expectation_{\muT}(\pBel_i(\varphi)@\act\Bigiven\,\racti) ~<~ 1-\vareps\delta' ~<~ 1-\vareps\delta$.
By \Cref{thm:mixedExpectationOfBeliefTrm} we obtain that
$\muT\big(\varphi@\act\Bigiven\,\racti\big) \,=\, \expectation_{\muT}(\pBel_i(\varphi)@\act\Bigiven\,\racti) $. 
We thus obtain that: 
$$\muT\big(\varphi@\act\Bigiven\,\racti\big) \quad=\quad
\expectation_{\muT}(\pBel_i(\varphi)@\act\Bigiven\,\racti) \quad < \quad  
1-\vareps\delta  \quad.$$
%
It thus follows that if $\muT\big(\varphi@\act\Bigiven\,\racti\big) \ge 1-\vareps\delta$, then $\muT\big(\pBel_i(\varphi)@\act \ge 1-\vareps\, \given \racti\big) \ge 1-\delta$.
\end{proof}

\vspace{5mm}
We are now ready to prove:\\[1ex] \newline
\textbf{Corollary 7.2.}\quad 
{\it Let $\act$ be a proper action for agent~$i$ in a pps~$T$, and let~$\varphi$ be local-state independent of~$\act$ in~$T$. 
For all $\vareps\ge 0$,
\hspace{3mm}if \hspace{2mm} 
$\muT(\varphi@\act\Bigiven\,\racti)\ge1-\vareps^2$ 
\hspace{3mm}then \hspace{3mm}
$\muT\Big(\pBel_i(\varphi)@\act\,\ge\, 1-\vareps~ \Bigiven\, \racti \Big)~\ge~ 1-\vareps.$
}

\begin{proof}
The case of $\vareps=0$ is simply Lemma~\ref{lem:probability1}. For $\vareps=1$ the claim follows from the fact that~$\muT$ is a probability measure.
Finally, for $0<\vareps<1$, the claim is an instance of \Cref{thm:PSK} obtained by setting $\delta=\vareps$. 
\end{proof}

\end{document}